\documentclass[11pt]{article}
\pdfoutput=1
\usepackage[utf8]{inputenc}
\usepackage{hyperref,dcolumn,bbm,url,booktabs,braket,microtype,aliascnt,mathtools,mathrsfs,etoolbox,cancel,capt-of,algpseudocode,float,newfloat,amsmath,amssymb,color,amsthm,fullpage,cleveref}
\usepackage[T1]{fontenc}
\usepackage[USenglish]{babel}
\usepackage{ marvosym }
\usepackage{tabularx}
\usepackage{array, ltablex, multirow}
\usepackage{fancyvrb}
\usepackage{adjustbox}
\usepackage{subcaption}
\usepackage{hyperref}
\usepackage{enumitem}
\usepackage{amsmath, amssymb, braket, mathrsfs}
\usepackage{physics}

\makeatletter
\def\blx@bblversion{3.1}
\makeatother
\usepackage[style=numeric-comp,sorting=none,arxiv=abs,eprint=true]{biblatex}
\AtEveryBibitem{%
  \ifboolexpr{
    test {\iffieldundef{doi}}
    and test {\iffieldundef{journal}}
  }{}
  {%
    \clearfield{eprint}
    \clearfield{archivePrefix}
    \clearfield{primaryClass}
    \clearfield{url}
    \clearfield{urldate}
  }
}

\usepackage{graphicx}
\usepackage[margin=1in]{geometry}
\usepackage[percent]{overpic} % overlay labels in percentage coordinates
\usepackage{tikz}
\usepackage{authblk}
\usepackage{orcidlink}

\usepackage{xcolor}
\newtheorem{theorem}{Theorem}[section]

\newtheorem{definition}[theorem]{Definition}

\newtheorem{corollary}[theorem]{Corollary}
\newtheorem{conjecture}[theorem]{Conjecture}

\title{Heuristic Quantum Advantage with Peaked Circuits}
\author[1]{Hrant Gharibyan} 
\author[1]{Mohammed Zuhair Mullath} 
\author[1]{Nicholas E. Sherman} 
\author[1]{Vincent P. Su} \author[1]{Hayk Tepanyan}
\author[2,3]{Yuxuan Zhang\orcidlink{0000-0001-5477-8924}} 
\affil[1]{BlueQubit

San Francisco, CA 94105, USA}
\affil[2]{Department of Physics, University of Toronto

Toronto, ON M5S 1A7, Canada}
\affil[3]{Vector Institute for Artificial Intelligence

Toronto, ON M5G 0C6, Canada}

\date{\today}

\begin{document}

\maketitle

\begin{abstract}
We design and demonstrate heuristic quantum advantage with peaked circuits (HQAP circuits) on Quantinuum's System Model H2 quantum processor. Through extensive experimentation with state-of-the-art classical simulation strategies, we identify a clear gap between classical and quantum runtimes. Our largest instance involves all-to-all connectivity with 2000 two-qubit gates, which H2 can produce the target peaked bitstring directly in under 2 hours. Our extrapolations from leading classical simulation techniques such as tensor networks with belief propagation and Pauli path simulators indicate the same instance would take years on exascale systems (Frontier, Summit), suggesting a potentially exponential separation. This work marks an important milestone toward verifiable quantum advantage, as well as providing a useful benchmarking protocol for current utility-scale quantum hardware. We sketch our protocol for designing these circuits and provide extensive numerical results leading to our extrapolation estimates. 
Separate from our constructed HQAP circuits, we prove hardness on a decision problem involving generic peaked circuits. When both the input and output bitstrings of a peaked circuit are unknown, determining whether the circuit is peaked constitutes a QCMA-complete problem, meaning the problem remains hard even for a quantum polynomial-time machine under commonly accepted complexity assumptions. Inspired by this observation, we propose an application of the peaked circuits as a potentially quantum-safe encryption scheme~\cite{chen2016report,kumar2020post,joseph2022transitioning,dam2023survey}.
%suggesting an encryption scheme as a potential application of peaked circuits. %\VS{TODO group: double check this is the message we are going for} \hayk{I like the mention of QCMA-complete but need to make sure we dont make this claim about HQAP. Also would use the keywords "application of peaked circuits" somewhere. smth like "When both the input and output bitstrings are unknown, determining whether a circuit is peaked constitutes a QCMA-complete problem for random peaked circuits, suggesting encryption as a possible application of peaked circuits."}
%\hayk{can we put the word "application" in the phrasing Yuxuan?} 
We make our peaked circuits publicly available and invite the community to try additional methods to solve these circuits to see if this gap persists even with novel classical techniques.

\end{abstract}

\pagebreak
\tableofcontents

\section{Introduction}

Random circuit sampling (RCS)~\cite{Aaronson:2016guw,Bouland:2018bva,Hangleiter:2022ibu} has emerged as a primary benchmark for demonstrating quantum computational advantage, exemplified by Google's 2019 Sycamore experiment~\cite{Arute2019QuantumSupremacy} and later followed by USTC Zuchongzhi works~\cite{wu2021strong,zhu2022quantum}. However, recent large-scale demonstrations from Google's Willow~\cite{Morvan2024PhaseTransitions} and Quantinuum's H-2~\cite{PhysRevX.15.021052}, while likely beyond classical simulation, expose a fundamental limitation of random circuit sampling: verification through cross-entropy benchmarking becomes exponentially expensive at large scales, rendering direct validation impossible. 

We present a scalable method for demonstrating heuristic quantum advantage using peaked quantum circuits~\cite{aaronson_2024_peaked}, circuits engineered to produce specific output bitstrings with anomalously high probability compared to an truly random circuit but otherwise look random.
The verification protocol is straightforward: Alice constructs a peaked circuit knowing its peak bitstring, Bob executes the circuit with his quantum computer measuring in the computational basis, and Alice verifies by checking if Bob's output matches the known peak. No exponential classical computation is required.

Our contributions are threefold: (i) We develop scalable construction techniques for peaked circuits at larger qubit counts and depths using methods related to identity obfuscation~\cite{doi:10.1142/S0219749905001067,ji2009nonidentitycheckremainsqmacomplete,Mori:2023ceu}.  (ii) Through extensive benchmarking, we demonstrate that our circuits resist state-of-the-art classical simulation methods including matrix product states~\cite{ORUS2014117,RevModPhys.93.045003, PhysRevResearch.3.023073, Biamonte_2017}, tensor networks with belief propagation~\cite{PhysRevResearch.3.023073, 10.21468/SciPostPhys.15.6.222}, and Pauli path simulators~\cite{Rall2019PauliPropagation,Aharonov2023NoisyRCS,Schuster:2024jds,doi:10.1126/sciadv.adk4321,Begusic:2023owa,PRXQuantum.6.020302,Gharibyan:2025ldn,Rudolph:2025gyq}. (iii) We motivate this construction further with hardness results for a decision problem involving general peaked circuits. In particular, we prove that determining whether an unknown circuit is peaked constitutes a QCMA-complete problem. Note that this is separate from our definition of ``solving'' our HQAP circuits, where the input string is fixed to the all zeros state.

Peaked circuits broadly have mounting evidence for their hardness. 
Simulating random peaked circuits is hard in the average case, analogous to RCS, particularly for circuits generated via post-selection \cite{Zhang:2025tuf} and error correction codes~\cite{deshpande2025peaked}. Our theoretical results in this work further add to the literature on hardness when the quantum circuit is peaked from a specific input bitstring to a particular output bitstring.
However, our primary contribution lies in
demonstrating the empirical difficulty of classically simulating circuits that not only admit efficient construction, but can also be implemented on existing quantum hardware. The relationship between our practical construction and the broader ensemble of peaked circuits remains an open question. In the absence of this theoretical connection, we refrain from asserting strong classical intractability results, instead characterizing the observed quantum-classical separation as a \textit{heuristic} quantum advantage.

The rest of the paper is organized as follows. 
In Section 2, we summarize the empirical runtime gap between quantum and classical resources needed to solve our peaked circuits. 
In Section 3, we describe briefly the techniques used to construct these peaked circuits. 
In Section 4, we provide additional background and details on the classical methods we implement for comparison. 
In section 5, we study a potential practical application of peaked circuits as a post-quantum cryptographic primitive. 
We conclude in Section 6 and invite others to test other ways of finding the peak in Section 7.

\section{A Gap in Classical and Quantum Runtime for Solving Peaked Circuits}

We start with our main result, which is an empirical gap in the time it takes to solve peaked circuits that we have constructed. Though our methods can generalize to even larger circuits, we focus on what is achievable on existing hardware, in particular Quantinuum's H-2 processor~\cite{Moses_2023} with 56 qubits and any-to-any connectivity. We create peaked circuits with 56 qubits and two-qubit gate counts ranging from 200 up to 2000. The quantum run time is approximately a few hours, while classical methods fail to extract the correct answer at around 700 RZZ gate circuits with reasonable computational resources. We estimate the largest circuits would require millions of GPU hours.

From Ref.~\cite{aaronson_2024_peaked}, we take the following definition.

\noindent\textbf{Definition 1.1} (Peaked Circuit). \textit{Given $\delta \in (0,1]$, we call the unitary $C$ $\delta$-peaked if:}
\[
    \max_{s \in \{0,1\}^N} \, |\langle s | C | 0^N \rangle|^2 \geq \delta
\]
\textit{with a corresponding peak weight $\delta_\textbf{s} \equiv |\langle \textbf{s} | C | 0^N \rangle|^2$.} where $\textbf{s}$ is the peaked bitstring. To ``solve'' a peaked circuit means given $C$, e.g. a circuit description in terms of gates, one correctly identifies the peaked bitstring $\textbf{s}$. For sufficiently small systems, one can compute the full statevector and compute the maximal amplitude squared but this clearly does not scale well with system size. With access to a quantum computer, one can simply run the circuit an $O(1/\delta_\textbf{s})$ times to observe a signal. Of course, in the presence of noise, one may encounter an overhead in the required number of shots.

\begin{figure}[H]
  \centering

  % Left column (note the % at line end to kill inter-word space)
  \begin{minipage}[t]{0.5\textwidth}
    \vspace{0pt}\centering
    \begin{tikzpicture}
      \node[inner sep=0] (A) {\includegraphics[width=\linewidth,height=\linewidth,keepaspectratio]{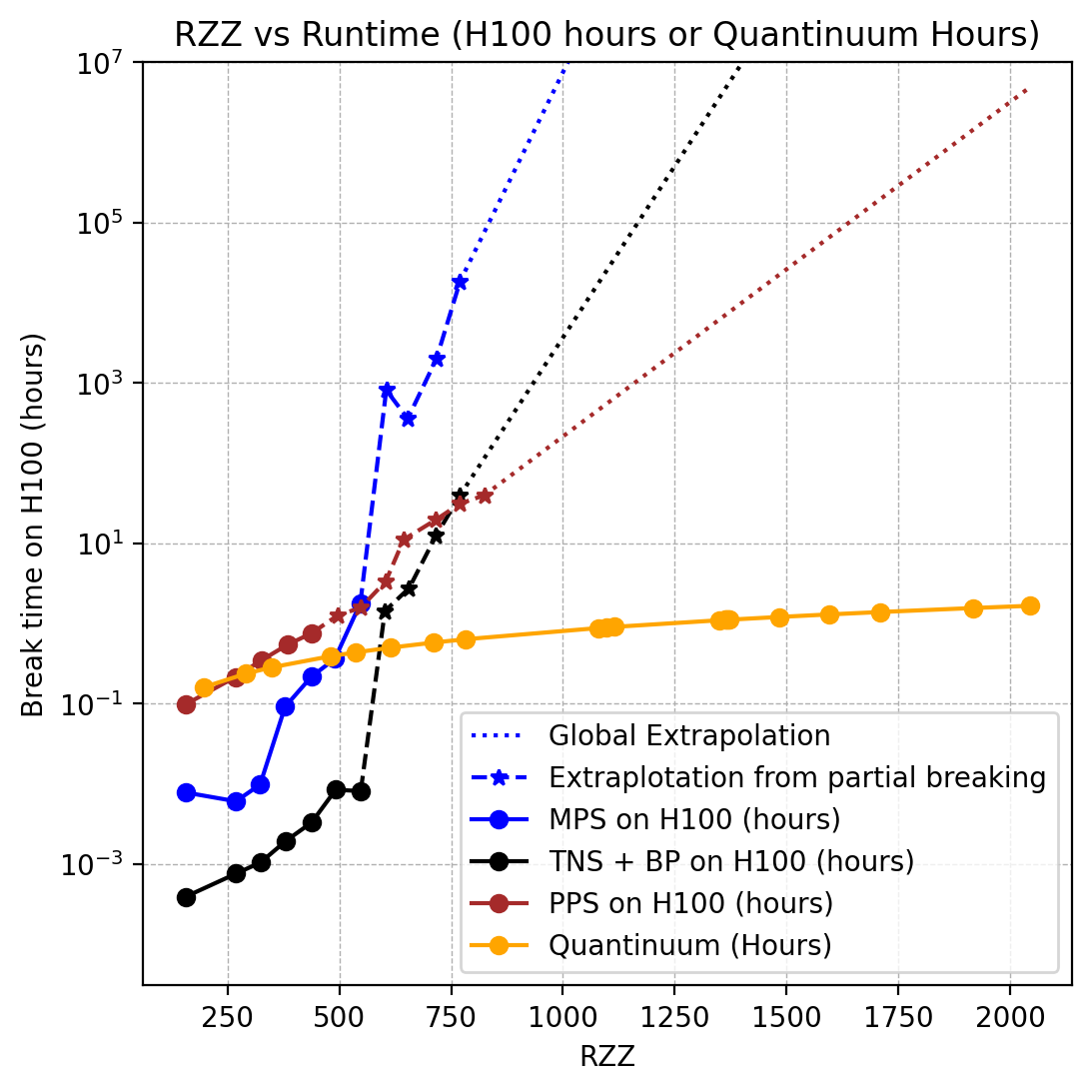}};
      \node[anchor=north west,xshift=-2pt,yshift=0pt] at (A.north west) {\small\bfseries a)};
    \end{tikzpicture}
  \end{minipage}\hfill%  <-- keep this % 
  % Right column
  \begin{minipage}[t]{0.5\textwidth}
    \vspace{0pt}\centering

    \begin{tikzpicture}
      \node[inner sep=0] (B) {\includegraphics[width=\linewidth,height=.5\linewidth,keepaspectratio]{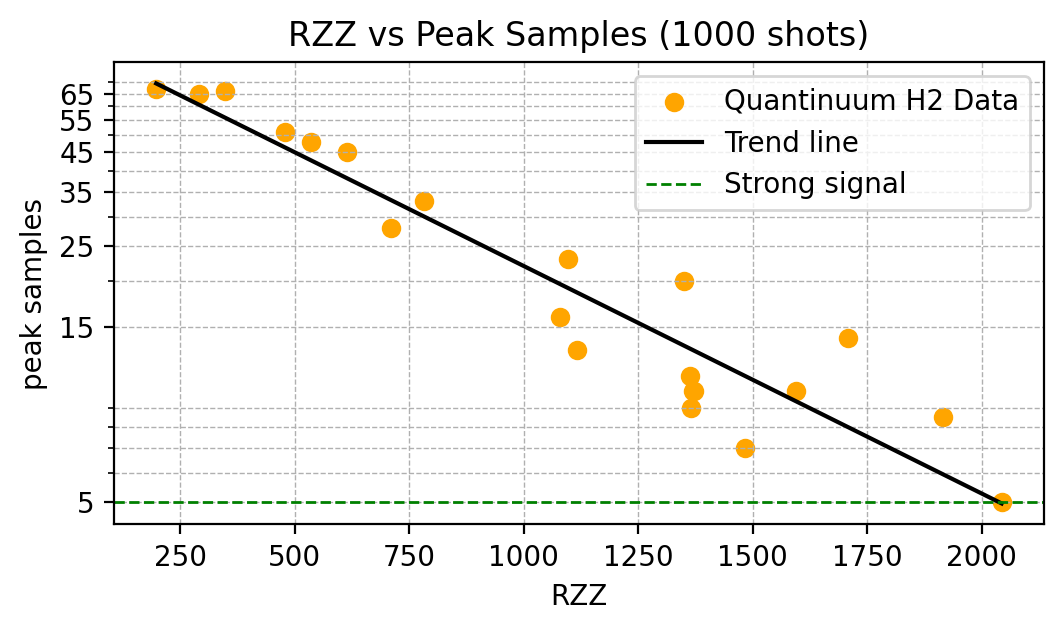}};
      \node[anchor=north west,xshift=-4pt,yshift=2pt] at (B.north west) {\small\bfseries b)};
    \end{tikzpicture}

    % \vspace{0.6ex}

    \begin{tikzpicture}
      \node[inner sep=0] (C) {\includegraphics[width=\linewidth,height=.5\linewidth,keepaspectratio]{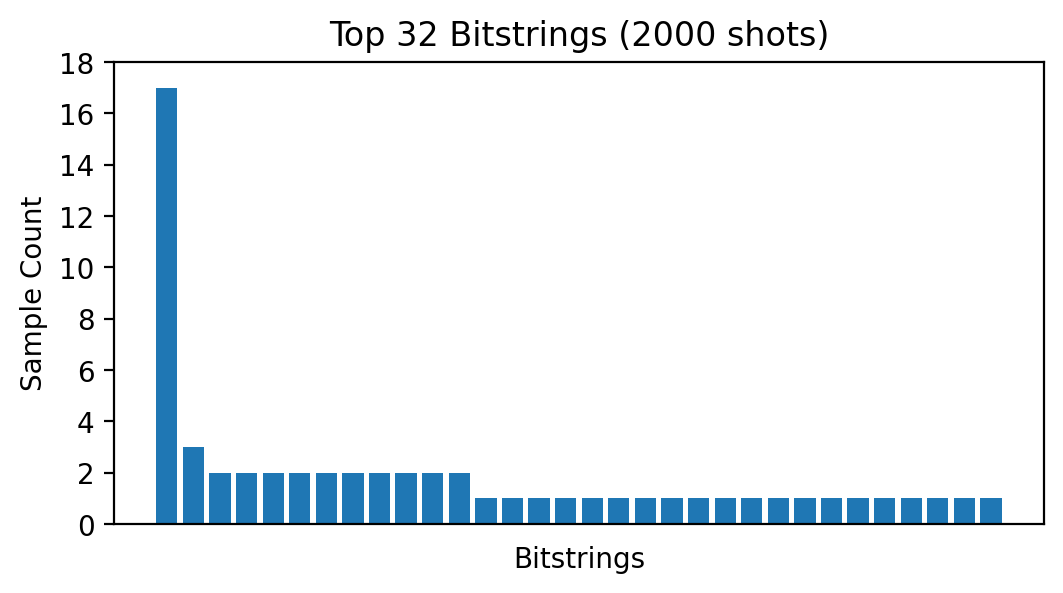}};
      \node[anchor=north west,xshift=-4pt,yshift=2pt] at (C.north west) {\small\bfseries c)};
    \end{tikzpicture}
  \end{minipage}

  \caption{\textbf{Gap in Classical and Quantum Runtimes for Solving 56 qubit HQAP Circuits}. \textbf{a)} We compare classical techniques (MPS, PPS and TNS+BP) with quantum techniques for solving 56 qubit, 10\% peaked circuits with all-to-all connectivity. As number of RZZs increases, classical methods require estimated millions of Nvidia H100 GPU hours, while the quantum device requires less than 2 hours. More details on the classical methods can be found in section \ref{sec:classical_methods} \textbf{b)} The performance of the quantum device drops with the growing number of RZZs in the HQAP circuit as expected, with still strong signal of 5/1000 (correct) peak samples for a 2044 RZZ circuit. \textbf{c)} Actual results from a single run on Quantinuum H2 for a 1917 RZZ peaked circuit, with 17/2000 correct peaked bitstrings. To our knowledge no other device in the world can find the peak for HQAP circuits of this scale and connectivity. }
  \label{fig:runtimes}
\end{figure}
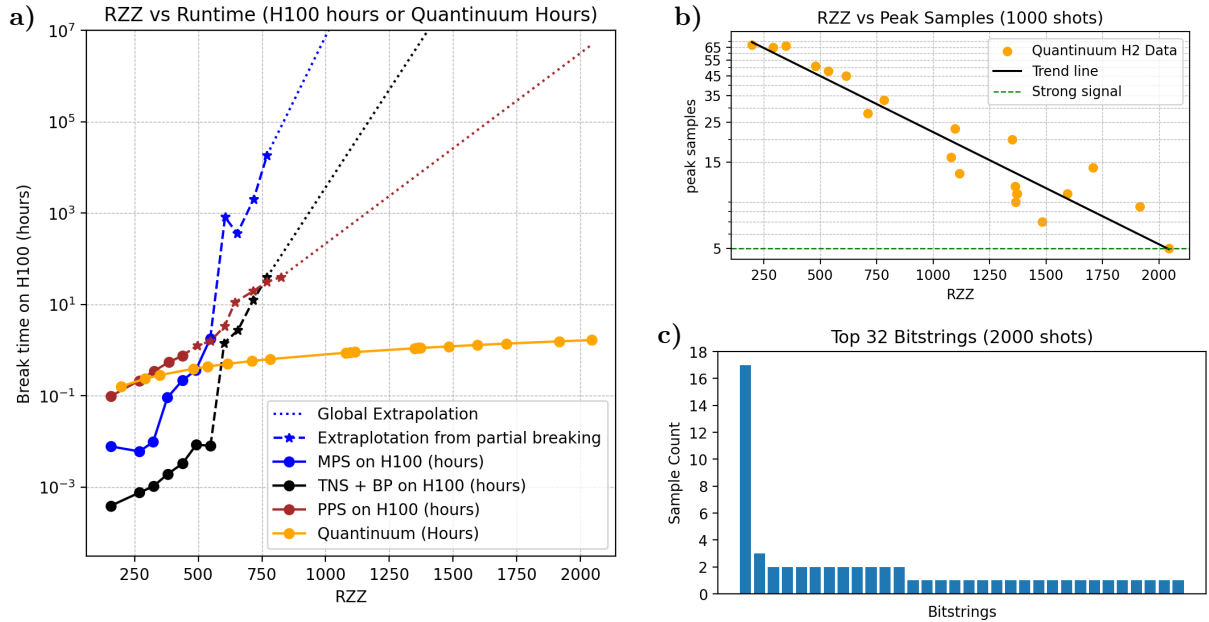

The main technical contribution of this paper is to create trainable circuits which go beyond the qubit count and depth of the results trained in Aaronson and Zhang~\cite{aaronson_2024_peaked}. Additionally, while the method they found had $\delta_\textbf{s}$ which seemed nearly universal based on their training strategy, we have greater flexibility in tuning the peak weight. We call these HQAP (Heuristic Quantum Advantage Peaked) circuits as they demonstrate a large heuristic gap between classical and quantum runtimes for solving them.
More on the protocol of how we generate HQAP circuits can be found in Section~\ref{sec:peaked_protocol}.

The size of the peak presents an interesting trade-off. Increasing $\delta_\textbf{s}$ means fewer shots would be required, reducing the quantum runtime. On the other hand, one may worry that tuning it to be too large may lead to classically exploitable structure in the circuit description. This is an open question, and we invite curious researchers to participate in our open challenge (see Section~\ref{sec:open-peaked-challenge}). Without that exploitable structure, one is forced to deal with generic approximate classical simulation methods. Some of the leading ones include Matrix Product State (MPS) simulators \cite{SCHOLLWOCK201196}, tensor network with belief propagation (TN + BP) \cite{PhysRevResearch.3.023073, 10.21468/SciPostPhys.15.6.222}, and Pauli path simulators (PPS)~\cite{Rall2019PauliPropagation,Aharonov2023NoisyRCS,Schuster:2024jds,doi:10.1126/sciadv.adk4321,Begusic:2023owa,PRXQuantum.6.020302,Gharibyan:2025ldn,Rudolph:2025gyq}. For more details on our implementations and attack strategies, see Section~\ref{sec:classical_methods}.

With the explicit goal of constructing the largest possible separation in classical and quantum runtimes, we construct HQAP circuits that are designed to fit within the tolerance of existing quantum hardware and find an extreme gap. The largest such instance features 56 qubits, 2044 two-qubit gates, all-to-all connectivity and a peak weight $\delta_\textbf{s} \approx 0.1$. With these parameters, Quantinuum's H2 device can consistently detect the peak in under 2 hours with 1000 shots, assuming an average speed of 3000 HQC/hour. 

On the classical side, we find that state-of-the-art simulation methods are only able to solve the circuits up to $\sim 700$ two-qubit gates within 10 hours. For HQAP circuits smaller than that, we find a nice extrapolation of runtime with two-qubit gate count. Assuming such larger computations were even feasible (e.g. with infinite memory and millions of CPUs or GPUs), our benchmarks estimate that solving the largest circuit we ran on Quantinuum would take years even on a Frontier scale supercomputer. More details on the classical methods used can be found in the section \ref{sec:classical_methods}. 

Again, though these circuits lack any theoretical guarantees, we claim that this is a form of \textit{heuristic} quantum advantage, where there is an efficient verification protocol for the quantum device, and as far as we know, classical methods fail to produce the same answer.

\section{HQAP Circuit Protocol}\label{sec:peaked_protocol}

The protocol to design the peaked circuits in this work is based on the construction discussed in Ref. \cite{aaronson_2024_peaked}. In that work, peaked circuits are composed of two main parts, a random quantum circuit $R$, and a variationally parmaterized circuit $P$ that serves as a \textit{peaking layer}. Using gradient descent, the parameters of $P$ are trained so that the combined circuit, $R$ followed by $P$, results in a peak weight on a particular bitstring. When the depth of $P$ exceeds the depth of $R$, the peaking can be done trivially and to a peak weight of 1.\footnote{Assuming that the connectivity of $P$ matches that of $R^\dagger$} 
The intuition for using this structure is that $R$ mimics the structure of random quantum circuits, and thus may be hard to simulate classically, even if the final state is relatively simple to describe. This protocol has two main challenges when scaling to larger circuits. First, as the circuits become larger, the presence of barren plateaus~\cite{McClean_2018-barren} makes training difficult in general. Second, training the $P$ requires simulating the full circuit, which can become excessively expensive or even impossible in practice. These challenges limit the theoretical bound on the gap between classical and quantum runtimes possible.

Here, we improve upon this construction by extending both the depth and the width of such circuits, primarily enabled by tensor network methods. By utilizing generic tensor contraction functionality as in Quimb~\cite{Gray2018quimb,Gray:2020cah}, we can train $R$ and $P$ as in Ref.~\cite{aaronson_2024_peaked} for large qubit number as long as they are sufficiently shallow.  Such shallow peaked circuits are known to be easy to solve~\cite{bravyi2024classical}. To increase the depth, we introduce a large, obfuscated identity block between $R$ and $P$. In the worst case, it has been shown that checking whether a quantum circuit is equivalent to the identity is a QMA-complete problem~\cite{doi:10.1142/S0219749905001067}, even when the circuit depth is shallow~\cite{ji2009non}. This indicates that identity obfuscation could be a feasible direction for generating peaked circuits that cannot be simplified by an attacker.
%However, 
In practice, to make the structure harder to reverse engineer, we employ a technique we call \textit{Tensor Patch Optimization}, which serves two purposes, patch sweeping and patch masking, which will be explained shortly. 
\paragraph{A word on notation.}
To enhance clarity in representing the sequential application of quantum operations, we introduce the left-to-right composition operator \( \triangleright \), defined by:
\[
U_1 \triangleright U_2 \triangleright U_3 := U_3 U_2 U_1
\]
That is, \( U_1 \) is applied first, followed by \( U_2 \), then \( U_3 \). This reverses the standard right-to-left ordering of matrix multiplication to improve readability when reasoning about temporal circuit structure.

\begin{figure}
    \centering
    \includegraphics[width=0.5\linewidth]{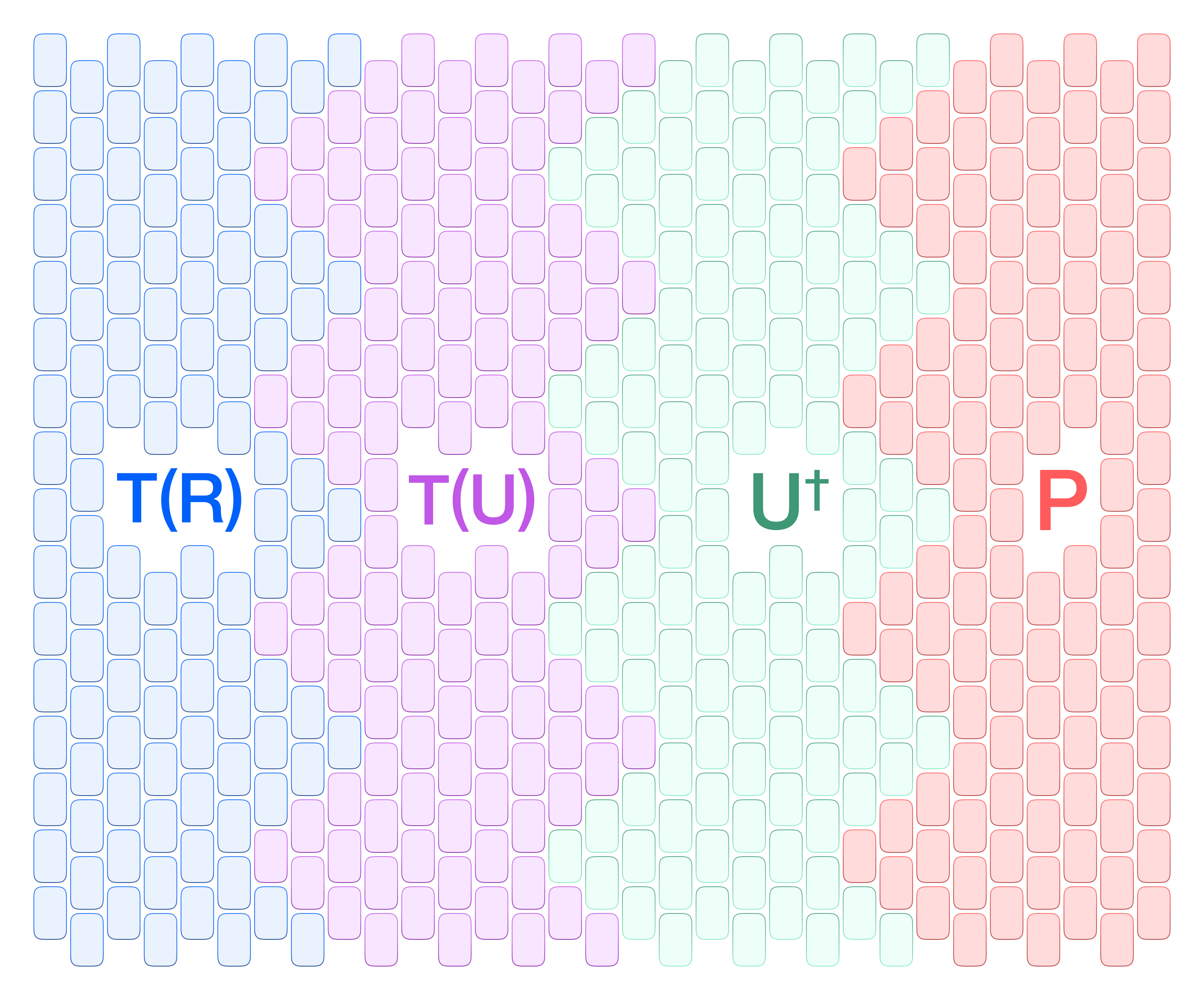}
    \caption{Visual schematic of HQAP building. We include details of the circuit manipulations used to scramble the circuit structure.}\label{fig:peak_protocol}
\end{figure}

Note that the identity block $U\triangleright U^{\dagger}$ implemented naively as gate-by-gate inversion is easy to detect, so we obfuscate this exact structure. The full protocol to generate peaked circuits is illustrated in Figure \ref{fig:peak_protocol}, and is summarized as
\begin{itemize}
    \item Train a shallow peaked circuit $R \triangleright P$ using the method in Ref. \cite{aaronson_2024_peaked}. 
    \item Insert an identity of the form $U\triangleright U^{\dagger}$ between $R$ and $P$.
    \item Apply sweeping to variational circuit parameters.
    \item Apply masking to modify the connectivity and gate structure.
    \item Apply swap transformations throughout the circuit.
\end{itemize}
The sweeping and masking, whose definition we will return to momentarily, can be done modularly throughout the circuit. However, the swapping must be coordinated globally to preserve the action of the unitary on the initial state. Let us denote the collective action of swapping, sweeping and masking as $T$, then the final circuit has the form $T[R] \triangleright T[U] \triangleright U^{\dagger} \triangleright P$. Note that the last few circuit deformation techniques can be applied in any order. Further details on the obfuscation strategies in the protocol are discussed in the following sections.

For our construction there exist, in principle, simplifications that may reduce the simulation requirement. With knowledge of how they were constructed and subsequently modified, one can trace out the simplifications. However, we expect that the series of changes are likely hard to reverse. We have internally tested a number of strategies for attacking these circuits, which led us to the current set of techniques we describe below. Abstractly, these can be thought of as methods to make transpilation or finding an efficient contraction path of the quantum circuit more difficult. We liken our scrambling techniques (swaps, sweeps, masks) to how hash functions build one-way transformations. Our expectation is that while its non-trivial to prove how our scrambling techniques make the final circuit irreversible to the original circuit, a lack of successful attempts to crack them will eventually build conviction that our scrambling is in practice irreversible. That would be analogous to how hash functions like SHA-1 and SHA-2 build irreversible 1-way protocols based on many small iterations. Similarly, SHA-1 and SHA-2 lack rigorous proof of irreversibility. However, the huge body of evidence shows that for all practical purposes they are hard to reverse.
For this same reason, we invite everyone to try to unscramble and break our HQAP circuits, see Section \ref{sec:open-peaked-challenge}.

\subsection{Tensor Patch Optimization}\label{sec:sweeping}
The idea of tensor patch optimization is to take a circuit and train a second variational circuit which matches the unitary action of the first. 
More procedurally, one can isolate subcircuits, or subsets of gates, and by finding an approximate substitute, locally deform them. A similar idea to this type of obfuscation was discussed in Ref.~\cite{Mori:2023ceu}, there called circuit unoptimization. These subcircuits are unitary and can generally be treated as tensors $\mathcal{T}_i$. The mathematical idea is to introduce a new trainable patch $\tilde{\mathcal{T}}(\theta)$ which is trained to mimic the original patch by taking the trace fidelity as a loss function 
\begin{equation}
    \mathcal{L}(\theta) = 1 -\frac{\left|{\rm Tr}[\mathcal{T}_i^{\dagger}\tilde{\mathcal{T}}(\theta)]\right|}{d},
\end{equation}
where $d$ is the dimension of the full unitary $\mathcal{T}_i$. Since $\mathcal{T}_i$ and $\tilde{\mathcal{T}}(\theta)$ are unitary, this loss function is always in the interval $[0, 1]$, and is $0$ if and only if $\tilde{\mathcal{T}}(\theta) = \mathcal{T}_i$ up to a global phase. 

Though mathematically similar problems, there are two ways we employ this technique. The first we refer to as angle \textbf{sweeping}, where the variational patch has the same gates as the original patch, but the parameters are first kicked away from their original values and then trained to increase the fidelity.
The second we call \textbf{masking}, where the trained patch actually has a different subcircuit structure, where the set of gates may be different or acting in a different order. Both offer a way to obfuscate structure that could potentially be used to simplify the circuit.

For the purposes of making the circuit harder to simulate, it can actually be beneficial to introduce a lossy approximation to the original patch. For example, when the loss function is small, but non-zero, one can no longer find an equivalence, thus making the transpiler's job harder.

\subsection{Swap Transformations}\label{sec:permutation}

The second obfuscation technique we employ is to apply a permutation enabled by a series of simple swap transformations between two qubit wires. Let \( \sigma^{(i \leftrightarrow j)} \) denote the unitary operation that swaps the quantum states of qubits \(i\) and \(j\). We define the action of this transformation on a circuit \(X\) as

\begin{equation}
    \sigma^{(i \leftrightarrow j)}[X] := \sigma^{(i \leftrightarrow j)\dagger} \, X_{\sigma} \, \sigma^{(i \leftrightarrow j)},
\end{equation}

where \( X_{\sigma} \) denotes the circuit \(X\) with its qubit indices relabeled under the transposition \( (i \leftrightarrow j) \). This can be interpreted either as a physical SWAP gate or as a virtual relabeling of the wire indices, depending on the context (e.g., compilation vs. simulation) 

These swaps can be interspersed throughout the circuit as long as one is careful about appropriately implementing the swaps and relabeling throughout. Such techniques make it difficult to design tensor network ansatzes that retain small bond dimension. Additionally, it is useful to note that the initial state is invariant under swaps, allowing all swaps at the beginning of the final transformed circuit to be ignored.

\section{Benchmarking Classical Approaches for Solving Peaked Circuits}\label{sec:classical_methods}

In the context of quantum advantage, it is increasingly understood that compelling demonstrations must not only exploit quantum computational power but also move beyond the reach of classical simulation. For sampling problems, this often means operating at a scale where even the most advanced classical techniques become infeasible due to computational or memory constraints. In particular, state-vector simulations, which scale exponentially in the number of qubits, are ruled out once the qubit count exceeds approximately 40. To test the classical hardness of our HQAP circuit constructions, we focus on scalable attack strategies that aim to approximately reproduce key features of the output distribution.

This section is broadly divided into three parts.
We start by briefly mentioning our primary attack strategy, looking at marginal or single qubit expectation values. This, in general, may not always converge to the correct answer, but is a general purpose approach that is simple to understand and can work when there is sufficient weight in the peak bitstring. In the case of examples we have run on quantum hardware, where we want the peak to be large enough to be detectable, this is a reasonable strategy. 

We subsequently explain the state-of-the-art classical methods used to implement this strategy. Details on simulation techniques are included. Because they fail to correctly output the target bitstring on moderately sized circuits with hundreds of two-qubit gates, we explain how our runtime extrapolations to the largest HQAP circuit sizes of $\sim2000$ two-qubit gates are calculated.

Importantly, these benchmarks of scaling generally serve as a lower bound on the time needed to use these classical methods as an attacker. Because we built the circuits to encode a particular bitstring, we are able to search for the minimal computational requirements needed to reproduce that answer. However, from the perspective of the attacker, these methods do not readily come with convergence guarantees, apart from doing full exact simulation. 

To give a sense of the classical computational resources used to reproduce these results, we spent thousands of A100 and H100 GPU hours and tens of thousands of CPU hours to solve the peaked circuits with the best known classical techniques. For statevector simulation we have been using qsim by Google on CPUs \cite{quantum_ai_team_and_collaborators_2020_4023103} and cuQuantum by Nvidia on GPUs \cite{Bayraktar2023cuQuantum}. For MPS, Tensor Network and Belief Propagation we used the quimb library on H100 GPUs \cite{Gray2018quimb,Gray:2020cah}. For Pauli path simulation we have been using a newly developed GPU implementation on H100 GPUs. All of the classical experiments have been orchestrated using the BlueQubit platform.

Finally, we explore solution strategies that exploit the structure of the circuit directly, bypassing the need for brute-force simulation. While our circuits do not yet guarantee security against all such attacks, we present  evidence of baseline structural resistance.

\subsection{Marginal attack strategies}\label{sec:z_attack}

Sampling from the output distribution of a general quantum circuit is widely believed to be classically intractable. 
In the average case, even approximate sampling to small total variation distance is believed to be classically intractable, and under plausible assumptions such as the non-collapse of the polynomial hierarchy, this forms the basis of most near-term quantum advantage proposals. For random peaked circuits, it has been shown that simulating the output distribution of peaked circuits to $1/\exp{\text {poly}(N)}$ precision is average case \#P hard; relaxing this precision to $1/\text {poly}(N)$, in the worst case the problem is still BQP-complete, and for any sequential simulator it is average case BQP-complete under plausible complexity assumptions~\cite{Zhang:2025tuf}. Similar hardness results have been observed in other types of peaked circuits~\cite{deshpande2025peaked}. Nevertheless, in this work there is a caveat one could leverage: our task is just to `identify the peaked strings', leaving room for spoofing. Although exact evaluation of expectation values of arbitrary circuits is computationally prohibitive, estimating approximate expectation values, particularly of low-weight observables or marginal distributions, is often significantly more tractable for classical algorithms, especially when the circuit depth is moderate or its structure is constrained, such as in circuits with shallow magic depth~\cite{zhang2025classical}.

The idea is relatively straightforward. If there is sufficient bias in the peak, then approximate expectation values, as long as they get the correct sign, can be used to attack or reconstruct what the peak bitstring is. Of course, as circuits get deeper, even getting approximate expectation values can be difficult. Indeed, for our largest circuits, our classical methods cannot solve them. Nonetheless, we introduce some of our notation that will be useful in describing their (lack of) success. 

For sufficiently large peak weight $\delta_\textbf{s}$, it is straightforward to observe that each of the bits will have marginal distributions $p(b_i)$ which are heavily biased towards the true value $s_i$. 
This is straightforwardly related to the $Z$ expectation value
\begin{equation}
    p(b_i=0) = \frac{1 + \langle Z_i \rangle}{2}.
\end{equation}
Thus, by coming up with estimates for $\langle Z_i\rangle$, e.g. with approximate tensor network methods, one can construct a candidate bitstring $\hat{s}$ by looking at the sign of each $\langle Z_i\rangle$. We refer to this as the marginal attack. From an attacker's perspective, this is not guaranteed to work. Indeed, one can manipulate the peak weight and the distribution so that marginals, even if calculated correctly, would give the wrong result. However, for the purposes of the classical benchmarks we present here, we have constructed the circuits so that the marginals are correctly lined up with the target bitstring. 

To keep track of the accuracy, one can consider the overlap (or inverse hamming distance) $O=N-|s-\hat{s}|_1$, or similarly, the fraction of correct bits $R=O/N$. Note that for a random guess one expects $R= 0.5$, and $R=1$ is only achieved if ${\textbf{s}} = \hat{\textbf{s}}$.

Interestingly, this attack strategy can also be used when looking at empirical or approximately generated samples, and corresponds to majority vote on each qubit.

For the largest circuits, our experiments are dominated by noise with minimal or no signal at all $R\lesssim 0.5$. For this reason, we study circuits of different sizes and compute $R$. We then define $T_{\rm break}$ as the minimum time needed to achieve $R=1$, or a quantity such as $\chi_{\rm break}$ defined as the smallest bond-dimension needed to achieve $R=1$ for MPS and TNS simulations. Since the runtime and accuracy for simulating a circuit is fully characterized by the bond-dimension $\chi$, there is a direct relationship between $t_{\rm break}$ and $\chi_{\rm break}$.

One can play a similar game with access to multi-qubit observables, though these in general will not be any easier to calculate than single qubit observables. When dealing with samples, whether approximate or empirical, one can implement many classical post-processing schemes. They may consider correlations or do majority votes on subsets of bits rather than single bits. However, finding an optimal strategy when there are conflicts is also very difficult.

\subsection{Matrix Product State Simulations}\label{sec:MPS_results}
MPS has revolutionized the study of 1D quantum systems \cite{SCHOLLWOCK201196}. Its power for capturing and studying entanglement has thus become a standard tool in the kit of many quantum practitioners. MPS allows a dial between efficiency and expressibility through a simple parameter, the bond dimension $\chi$.

\begin{figure}
    \centering
    \includegraphics[width=0.7\linewidth]{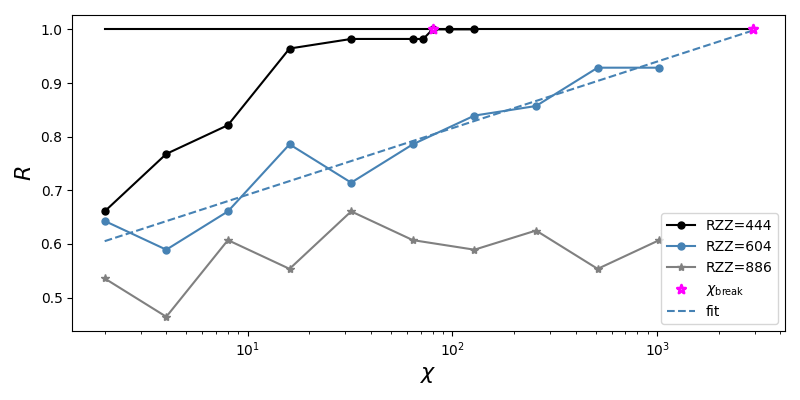}
    \caption{Estimating $\chi_{\rm break}$. By tracking the accuracy $R$ as a function of $\chi$, we find 3 cases. 
    The black curve is an example of a circuit where $\chi_{\rm break}$ is found exactly. For the blue curve, the $\chi_{\rm break}$ is not found exactly, and so we fit the data to find $\hat{R}(\chi)$ and extrapolate to estimate $\chi_{\rm break}$, and show the fit with a dashed line. Lastly, the gray curve is an example of a circuit our fitting protocol rejects when estimating $\chi_{\rm break}$. 
    \label{fig:mps_chi_scaling}}
\end{figure}

Our peaked circuits are explicitly designed with all-to-all connectivity. Thus, one would not expect MPS to be particularly well suited to simulate them. However, they serve as a baseline simulation method for understanding the classical difficulty. In particular because producing the correct the peak bitstring could be spoofed by a low rank MPS that is far from the true final state. One could worry that basic MPS simulation may crack the circuits. We find this naive method works up to a point. Beyond 600 two-qubit gates, however, MPS fails to predict the correct bitstring.

Because MPS works best for states with area law entanglement in 1D, the explicit transpilation of a circuit with long-range entangling gates to a chain layout can greatly affect the performance. In the worst case, a naive permutation of an otherwise local 1D chain can make MPS simulation struggle.
To define the ordering of the qubits, we look at the adjacency matrix for our quantum circuit and use the Reverse Cuthill-McKee (RCM) algorithm~\cite{RCM_graph_algo} to define a global permutation to our circuit. The RCM algorithm defines a reordering of an adjacency matrix to minimize the bandwidth of the matrix and hence the maximum distance between qubits that interact. 

\begin{figure}[htbp]
    \centering
    \includegraphics[width=\textwidth]{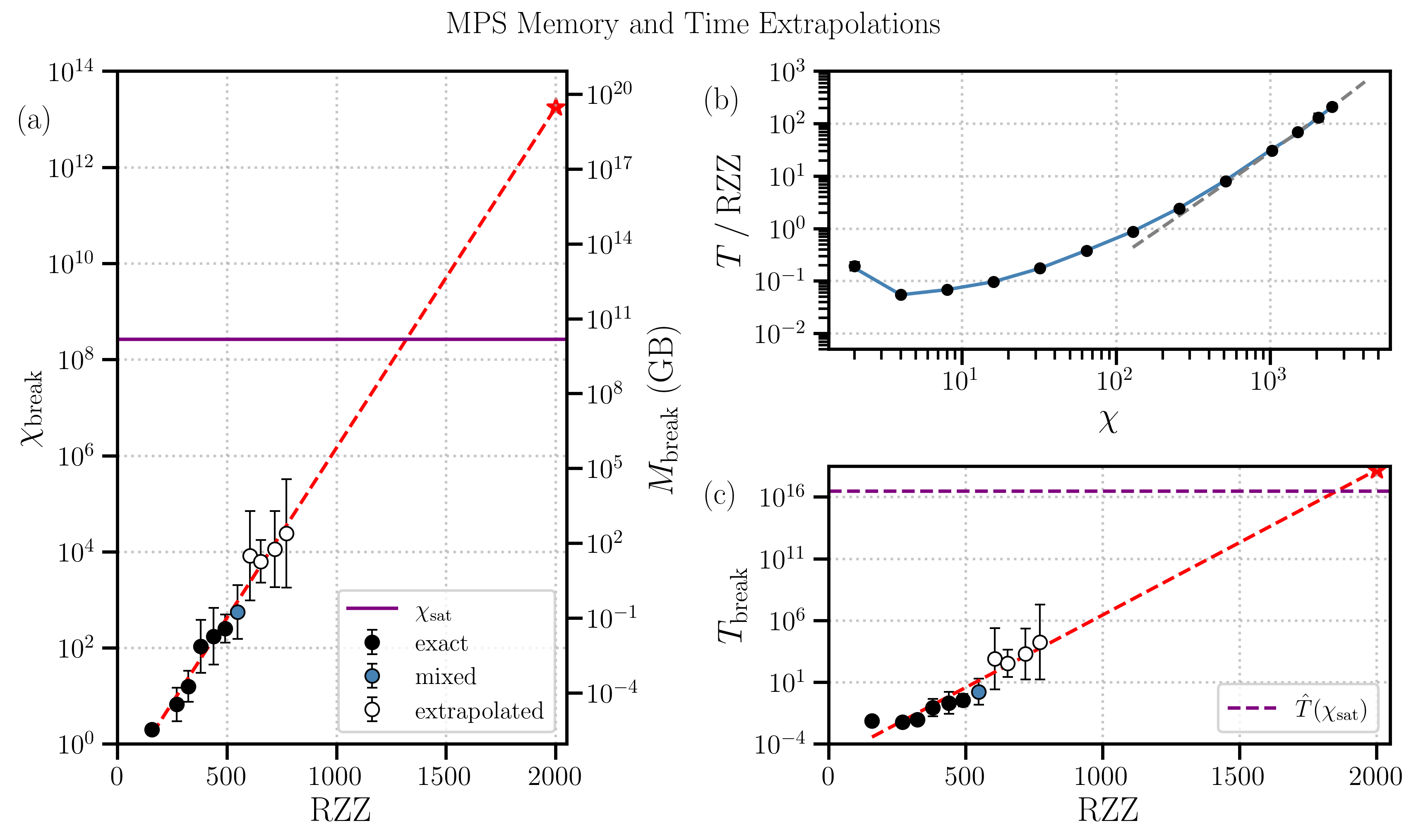}
    \caption{Benchmarking of bond dimension and time required to simulate peaked circuits of increasing depth with MPS. (a) Bond dimension vs two-qubit gate count. Black dots indicate MPS simulations that converged exactly, while white dots indicate circuit samples where the extrapolation procedure in Fig.~\ref{fig:mps_chi_scaling} was used. Blue dots indicate the transition where some circuits of that size were fully solved. The red dashed line shows the exponential fit, while the solid purple line shows the $\chi$ corresponding to exact simulation. (b) Gate simulation cost for increasing $\chi$. We study how the simulation time grows with increasing bond dimension, normalized by the number of gates. The large $\chi$ behavior is modeled by by a cubic function of $\chi$. (c) We plot the solution time $T_{\rm break}$ in hours as a function of circuit depth. Exact answers (black dots) can be read off, while white dots are extrapolated by using $\hat{\chi}_{\rm break}$ and the time conversion in (b). The red dashed line is an exponential fit for $T_{\rm break}$ compared with circuit depth. The purple dashed line is the estimate for the runtime of a depth 2000 circuit with $\chi_{\rm sat}$.}
    \label{fig:mps_combined_3panel}
\end{figure}

To produce extrapolation estimates for our largest circuits, we study the relationships between bond dimension, circuit depth, simulation time, and accuracy of the predicted answer. We define $\chi_{\rm break}$ as the minimum bond-dimension needed to yield the correct bitstring for a given circuit. For small circuits, this is found explicitly through simulation. For larger circuits, this can be extrapolated by looking at the trend of $R$, the correct fraction of bits, as a function of increasing $\chi$. 

For a given circuit, we perform MPS simulation with bond dimension increasing in powers of 2. If the circuit is cracked, we perform binary search to find $\chi_{\rm break}$. The maximum $\chi$ used in our experiments varies by circuit with $\chi=5000$ the largest bond-dimension used. If the circuit remains uncracked for all $\chi$ studied, we perform an automated fitting procedure to determine $\hat{R}(\chi)$. We assume a linear fit in $\log\chi$, weighted more heavily towards large values of $\chi$ and reject fits that yield a correlation coefficient of $r^2 < 0.8$, or those that show a negative slope. We show examples of $R(\chi)$ for three circuits, one which was cracked exactly, one which we successfully extrapolate $R(\chi_{\rm break})$, and one which we reject.

Our peaked circuit construction allows us to control the RZZ gate count. However, there is some variability in the resulting $\chi_{\rm break}$ due to the stochastic nature of the training. Thus, for a given circuit depth, we apply this extrapolation procedure to five instantiations, reporting both the average $\chi_{\rm break}$ as well as the standard deviation in the log. For accessible system sizes, we observe exponential scaling in the necessary bond dimension as a function of circuit depth, as might be expected for random circuits before the bond dimension saturates at its maximal value $\chi_{\rm sat}=2^{N/2} \simeq 10^8$~\cite{SCHOLLWOCK201196}. See Fig.~\ref{fig:mps_combined_3panel}. Interestingly, our extrapolated answer for $\chi_{\rm break}$ on our deepest circuits reaches beyond this $\chi_{\text max}$, giving evidence for the difficulty of MPS to crack our system sizes.

Ultimately, one cares not only about the memory requirements determined by $\chi$ but also the total runtime of the classical method. To this end, we also study the time $T$ to simulate these circuits as a function of $\chi$, normalized by the circuit depth. The runtime is dominated by the operations of applying a gate which scales as $\chi^2$, and the subsequent SVD which scales as $\chi^3$. We use this fact to fit our data to a functional form $T/RZZ \sim a\chi^2 + b\chi^3$, and show the fit in Figure \ref{fig:mps_combined_3panel}. 
Combined, these give a way to extrapolate the total runtime needed for our deepest circuits. Assuming one even has access to enough memory to run the simulation, we can use this relation to estimate $T_{\text break}$ based on $\chi_{\text sat}$ for our largest circuits.

\subsection{Tensor Network State and Belief Propagation Simulations}\label{sec:TNS_results}

Despite the success of MPS simulations in studying quantum dynamics in one-dimensional systems with local interactions, systems with higher-dimensions and long-range interactions have still remained a major hurdle in the field. A common approach has been to simply map such systems to quasi-one-dimensional systems at the cost of long-range interactions, as discussed in Section \ref{sec:MPS_results}. Other tensor network methods such as projected entangled-pair states (PEPS) \cite{verstraete2004renormalizationalgorithmsquantummanybody, RevModPhys.93.045003} have been popular for studying the many-body dynamics of two-dimensional lattices, such as for a square lattice, and easily extend to arbitrary geometries. However, extracting useful quantities, such as local observables, from the final state remains difficult in general due to the complexity of the resulting tensor network contraction. The complexity of local measurements disappears in MPS systems by using the gauge freedom of the MPS ansatz to bring an MPS state into a canonical form \cite{SCHOLLWOCK201196}. Once in this form, local observables can be computed by a tensor contraction involving $\mathcal{O}(1)$ site tensors rather than a full network contraction.

Recently, belief propagation (BP) has been proposed as a method to gauge a generic Tensor Network State (TNS) to address this problem \cite{PhysRevResearch.3.023073, 10.21468/SciPostPhys.15.6.222}. BP is a set of equations for \textit{message tensors} that can be solved self-consistently, whose fixed point defines a gauge transformation that brings the TNS into the Vidal gauge \cite{PhysRevLett.91.147902, PhysRevLett.93.040502, PhysRevLett.98.070201,PhysRevB.78.155117}. This procedure is known to be exact for tensor networks on a tree graph, and while approximate for loopy graphs, often leads to accurate results in practice.

Despite the benefits of TNS evolution and BP, this framework requires performing matrix decompositions on tensors, like QR or SVD, that scale as $\chi^{z}$ for geometries with a coordination number $z$. For two-dimensional systems with $z \simeq 4$, this scaling is manageable enough to lead to dramatic improvements in both accuracy and efficiency over traditional MPS-based simulations. However, for circuits with all-to-all connectivity, like the circuits studied here, the scaling becomes infeasible as $z\rightarrow N$ for a naive use of this method.

To remedy this, we define a simpler topology with a reduced coordination number, and then transpile the circuit to this simpler topology. This can be thought of as a generalization of what is done for MPS-based simulations, where the simpler topology is a chain. Once transpiled, we then perform the simulation using a generic TNS that matches the topology of the resulting circuit. With this increased flexibility comes the difficulty of picking a particular layout. There is a tradeoff between a simpler topology with a reduced $z$, allowing for more efficient tensor operations, and the increase in gate count from the added swap gates to perform the routing. 

The limitless options for both topology and bond dimension of the ansatz are further complicated by the use of BP at the end to compute the one-qubit marginals. For example, when using a $k$-regular graph with $3 \le k \le 6$, we found surprisingly that increasing $\chi$ for TNS states can actually lead to \textit{lower} accuracy with respect to the true expectation values. This is in sharp contrast with the MPS case where performance is expected to be monotonic in $\chi$. We believe this is due to the increase in long-range and loopy correlations throughout circuit simulation, which can spoil the accuracy of BP for computing the marginals \cite{PhysRevResearch.3.023073, 10.21468/SciPostPhys.15.6.222}. 

\begin{figure}
    \centering
    \includegraphics[width=0.55\linewidth]{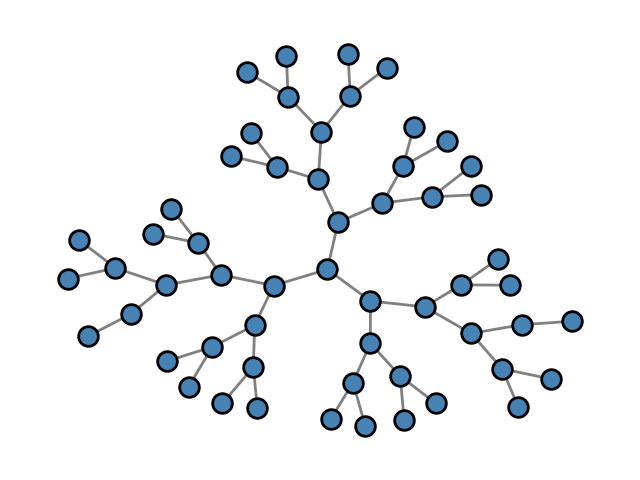}
    \caption{Example of the tree-like graph that we transpile our all-to-all circuits to before performing TNS simulations. The lack of loops lead to confidence in the convergence of BP for computing marginals in the tensor network state.}
    \label{fig:tns_tree}
\end{figure}

Given the stark effect of topology on the quality of the final answer, we found the most successful results using a tree tensor network. Notably, the lack of loops in the tree network satisfy a core assumption of BP. This led to the desirable property that increasing bond dimension did lead to more accurate marginals. In the tradeoff of gate transpilation and topology, we found reasonable results using a graph with three binary trees stemming from the root. All reported TNS simulations we benchmark use this graph. See Figure \ref{fig:tns_tree} for an illustration.

We perform the transpilation using the Qiskit package \cite{qiskit_paper}, and the gate evolution was done using the simple update algorithm \cite{PhysRevLett.101.090603, PhysRevB.99.195105} with the 'reduced tensor' method~\cite{10.21468/SciPostPhys.15.6.222}. Once we have an approximate TNS representing the final state of the circuit, we perform BP to bring the TNS into an approximate canonical form, allowing us to efficiently approximate $\langle Z_q\rangle$ for each qubit $q$. 

\begin{figure}[htbp]
    \centering
    \includegraphics[width=\textwidth]{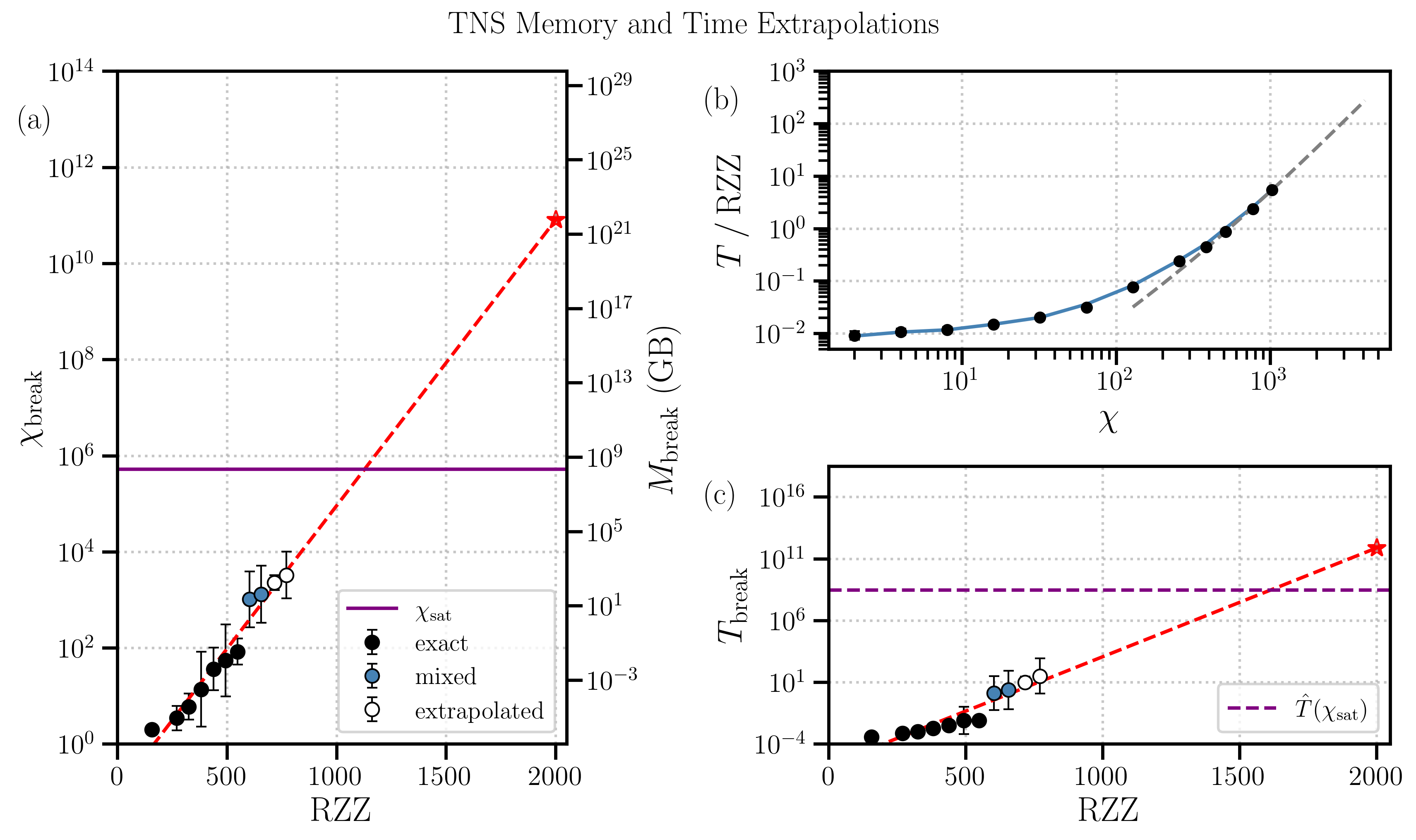}

    \caption{Memory and time extrapolations for TNS simulations. Notably, the TNS can crack circuits much more efficiently than in the case of MPS. However, we still find an exponential trend in the memory and time requirements. In the case of the network comprised of 3 binary trees, $\chi_{\rm sat}$ is determined by the size of the largest binary tree in Fig.~\ref{fig:tns_tree}. Though the $\chi_{\rm sat}$ is lower than in the MPS case, the required memory is comparable. Using the extrapolations with $\chi_{rm sat}$, we estimate that $T_{\rm break}$ for TNS simulations would take about $3 \times 10^8$ hours for a 2000 RZZ circuit, assuming the calculation would even fit in memory.}
    \label{fig:tns_combined_3panel}
\end{figure}

We then apply the marginal attack strategy discussed in Section \ref{sec:z_attack}. We also considered a greedy, iterative approach to generate a proposed bitstring $\hat{\textbf{s}}$. Rather than computing all marginals independently, one first selects the most biased marginal, and then fixes the tensor network based on the most likely bit. After fixing that bit, we re-run BP and compute the remaining marginals, conditioned on the outcome of the previous bits, and continue this process until all bits are fixed, defining $\hat{\textbf{s}}$. We then find $R$ of correct bits for both procedures, and use whichever method yielded a more accurate proposal.

Analogous to our MPS simulations, $\chi_{\rm break}$ is determined by $R(\chi_{\rm break}) = 1$. For sufficiently small circuits, we can compute $\chi_{\rm break}$ exactly, while for larger circuits we rely on fitting $\hat{R}(\chi)$ and extrapolate to approximate $\chi_{\rm break}$. In Figure \ref{fig:tns_combined_3panel}, we show $\chi_{\rm break}$ for peaked circuits with $N=56$ and a range of RZZ gates. 

Lastly, we create a conversion from $\chi_{\rm break} \rightarrow T_{\rm break}$ in the same way as we did for the MPS simulations described in Section \ref{sec:MPS_results}. The runtime is dominated by the operations of applying a gate which scales as $\chi^2$, and the subsequent SVD which scales as $\chi^3$. We use this fact to fit our data to a functional form $T/RZZ \sim a\chi^2 + b\chi^3$. This fit allows us to extrapolate for larger $\chi$ to estimate the runtime for larger circuits.
We note that from our fit of $T(\chi)$, the large $\chi$ scaling is only present for $\chi \gtrsim 256$. This leads to exponential scaling for $T_{\rm break} (RZZ)$ to only be visible for $RZZ \gtrsim 550$. Due to limited data in this regime, the estimated $T_{\rm break}$ should be treated as a rough estimate.

\subsection{Pauli Path Simulations}
Pauli Path Simulation (PPS) is a relatively new method to classicaly compute observable expectations through backward evolution of the circuit gates on the Pauli basis \cite{Rall2019PauliPropagation, Aharonov2023NoisyRCS, Schuster:2024jds,doi:10.1126/sciadv.adk4321,Begusic:2023owa,PRXQuantum.6.020302,Gharibyan:2025ldn,Rudolph:2025gyq}. We leverage a state of the art GPU-accelerated version of this to implement the marginal, or $\langle Z\rangle$, attack.
% It is commonly used to compute  observable expectation values, however there is an easy linear mapping between 1-qubit $\langle Z \rangle$ observables and 1-qubit marginals, so we can use this to perform an MPS-1 style attack. In principle, this approach could be extended to compute $k$-qubit marginals, potentially enabling a more general MPS-$k$ style attack; however, assessing the feasibility and effectiveness of such an extension is left to future work.

A counterintuitive finding of Ref.~\cite{Gharibyan:2025ldn}, is that the convergence patterns of observables as a function of coefficient  truncation parameter $\Delta$ is highly unpredictable. In the context of solving peaked circuits, this makes it difficult to predict, even for one $\langle Z_i \rangle$, what value of $\Delta$ will be sufficient to get an accurate expectation value. While intuitively one might expect smaller $\Delta$ should yield a more accurate answer, the deviations are non-monotonic, and thus a smaller $\Delta$ may actually be less accurate than a larger $\Delta$. 

The way we then choose to extrapolate the difficulty is by looking for convergence of the expectation value of each $\langle Z_i \rangle$ \textit{to the correct sign} and, hence, bit value. This assumption is then a \textit{lower bound} on the time it would take to crack with this strategy since it equips the attacker with the most information possible. As in the MPS case, we are using our knowledge of the true answer for extrapolation purposes. A priori, without the ground truth, confidence in the convergence can be very difficult to establish~\cite{Gharibyan:2025ldn}.

To be slightly more precise, we compute $\hat{\langle Z_i}\rangle$ for each qubit with geometrically decreasing truncation parameters $\Delta_j$. We check that i) the sign has converged for 3 successive $\Delta_j$ and is the correct value and ii) that the expectation value exceeds a certain fraction of the operator weight that survives the truncation. 

An interesting feature we notice when implementing this strategy is that there is generically a large gap between qubits that are easy to crack and those that are difficult. As can be seen from Figure \ref{fig:pps_vs_time} some qubits take more time to converge and crack, but the overall hardness grows exponentially with the number of 2-qubit gates in the circuit, as expected.

\begin{figure}[H]
  \centering
  \includegraphics[width=1.0\textwidth]{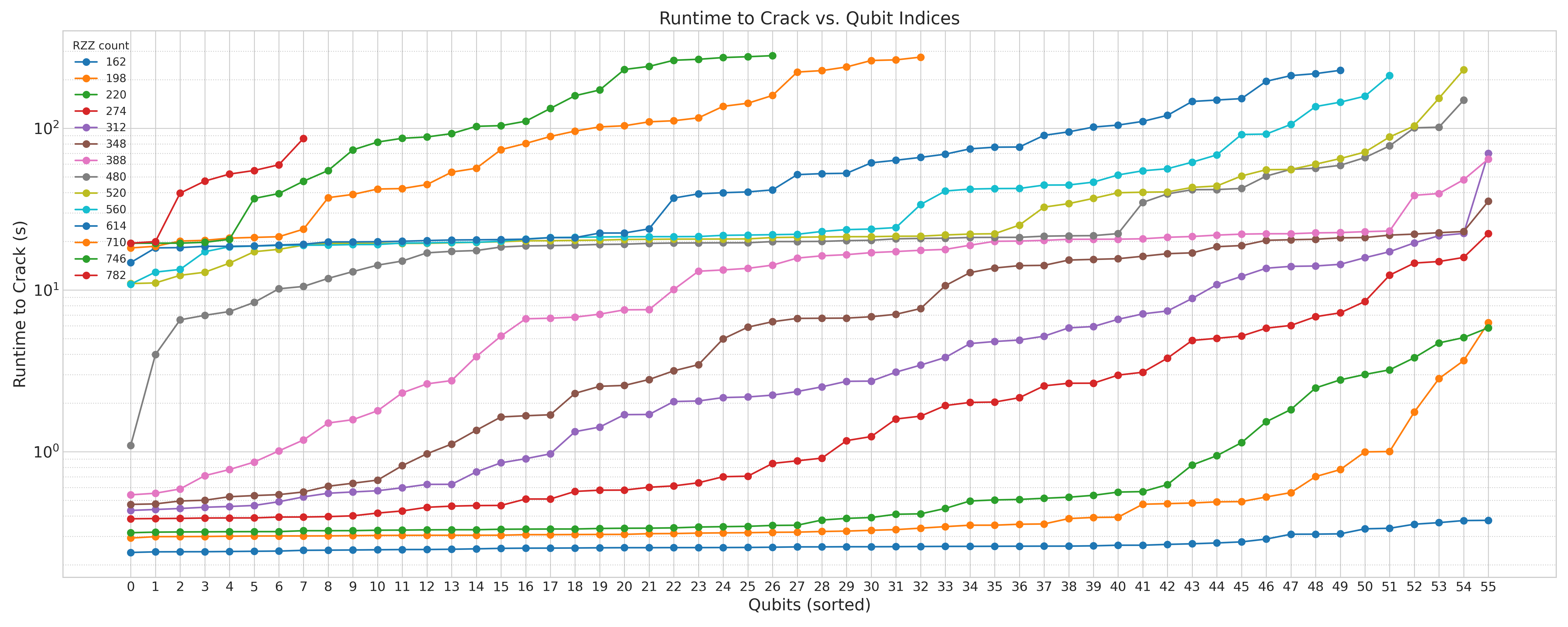}
  \caption{Required runtimes to crack individual bits of the peak bistring for 56-qubit HQAP circuit with varying RZZs going from 162 to 782 with PPS (standard coefficient truncation) on an H100 GPU. Qubit indices are ordered by difficulty (runtime to crack). We vary the truncation parameter $\Delta$ from $3.2 \times 10^{-3}$ to $2.5 \times 10^{-5}$ in powers of 2. }
  \label{fig:pps_vs_time}
\end{figure}

\begin{figure}[H]
  \centering
  \includegraphics[width=0.75\textwidth]{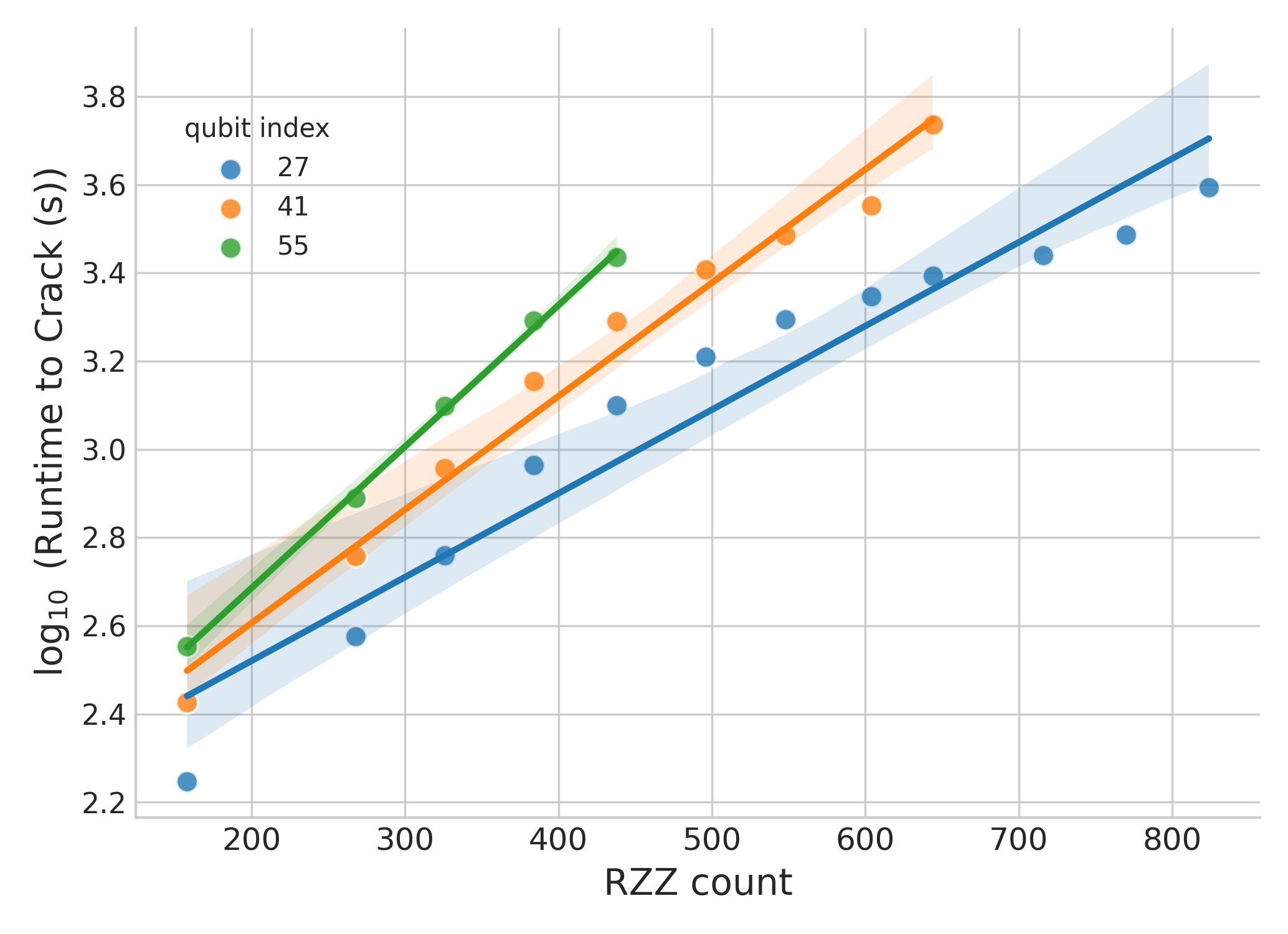}
  \caption{Plot showing exponential growth in runtime to crack bits of different difficulty levels for circuits with up to 824 RZZ gates. For this plot, we used the relative truncation method described in this section, and we tracked up to 1 billion Pauli terms, using $\sim$96 GB of memory on an H100 GPU. We note that solving qubits with index 55 amounts to fully solving the peaked circuit.}
  \label{fig:pps_runtime_pivot_plot}
\end{figure}

While PPS is very effective in shallower regimes and we can crack all 56 bits up to 388 RZZ HQAP circuit, memory very quickly becomes a constraint. Consequently, it was not feasible to explore truncation thresholds lower than $10^{-5}$ for these circuits, as the number of Paulis approached the memory limit ($\sim$1 billion Paulis). This is why the HQAP circuits after 388 RZZ have higher qubits missing from the plot in Figure \ref{fig:pps_vs_time}.

One of the issues we encountered while running PPS with standard coefficient truncation for larger peaked circuits (with RZZ-count $ > 600$) was that all the Pauli terms would get truncated midway through the simulation. The reason for this is the following: there are long stretches in the circuit where there is extensive branching and negligible merging of Paulis, resulting in rapid decay of all Pauli coefficients. 

We briefly outline a modification to the coefficient truncation method devised to mitigate this issue and (potentially) yield a non-zero expectation value at the end of each PPS run. The modification, which we call \textit{relative truncation}, is rather simple: after the $k$-th gate application, we drop only those Pauli strings $P$ whose coefficients $c_P$ satisfy $$|c_P| \leq \|O_k\|\Delta, $$ where $\|O_k\|$ is the normalized Frobenius norm of the evolved Pauli-sum after $k$ steps of the simulation. This simple modification dynamically lowers the effective truncation threshold from $\Delta$ to $\|O_k\|\Delta$, and when more and more Pauli terms are dropped, resulting in a decline in $\|O_k\|$, we end up truncating fewer Pauli terms than we would otherwise have. 

We also noticed an improvement in accuracy for the computed expectation values with this method, even after accounting for the fact that the new method utilizes more Pauli terms for a given threshold $\Delta$ compared to standard truncation. 

\begin{figure}[H]
  \centering
  \includegraphics[width=0.75\textwidth]{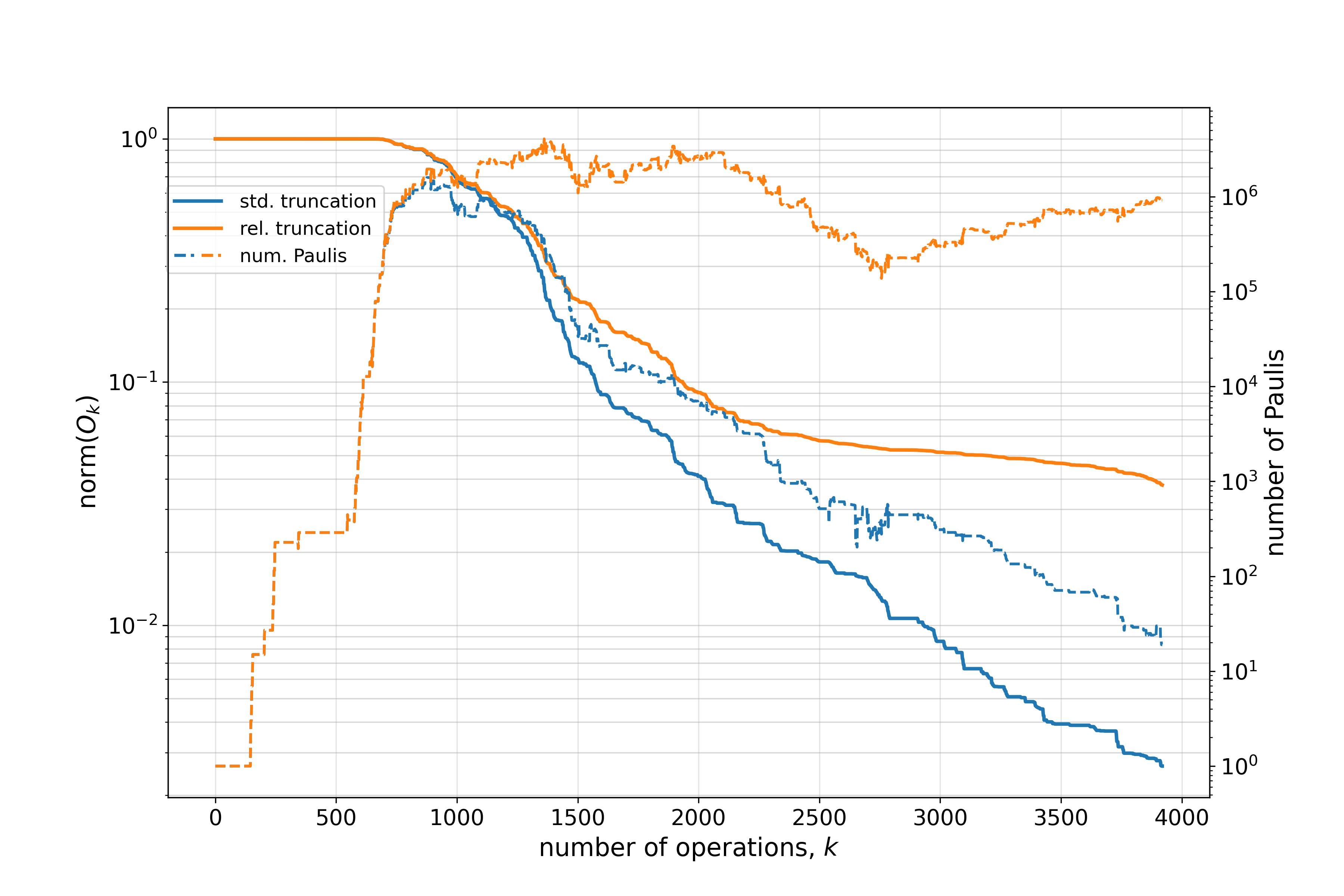}
  \caption{Comparison of change in observable norm $\|O_k\|$ and number of Pauli terms during the computation of $\langle Z_0 \rangle$ for both standard and relative coefficient truncation schemes for a circuit with 536 RZZ gates, and for threshold $\Delta = 2\times10^{-4}$. With standard coefficient truncation, we see that the number of Paulis decline rapidly in the latter half of the simulation, with only a couple of dozen terms remaining towards the end, whereas with relative truncation there are still around a million Pauli terms retained at the end, thereby improving the accuracy of the simulation.}
  \label{fig:pps_std_vs_rel_truncation}
\end{figure}

\subsection{Attacking the Circuit Structure}

We have examined the computational requirements for direct classical simulation of the circuit. While this provides a useful baseline, we recognize that targeted attacks exploiting specific circuit structures could potentially reduce these computational requirements. We present an initial analysis of several such attack vectors and discuss how the techniques introduced in Section 2 - swapping, sweeping, and masking - may offer some protection against them.

To illustrate these concepts, we analyze a simplified test case: a circuit constructed as $U \triangleright U^{\dagger}$ (a random unitary circuit followed by its adjoint). This allows us to verify that our methods pass a baseline level of resistance to these attacks, though we emphasize that making stronger guarantees on the difficulty remains an open question.

\textbf{Correlating angles}
For a random circuit that is comprised of multiple variational gates, the angles of gates in $U^\dagger$ will mirror those in the first half of the circuit. Thus, even though swaps were employed throughout, this would effectively allow someone to reverse engineer the resulting permutation. However, with tensor patch optimization used for the purpose of angle sweeping, we find that with a sufficient number of sweeps, one can hide the apparent correlations between different parts of the circuit. See Fig.~\ref{fig:sweep_angles}.

\begin{figure}[H]
  \centering
  \includegraphics[width=.9\textwidth]{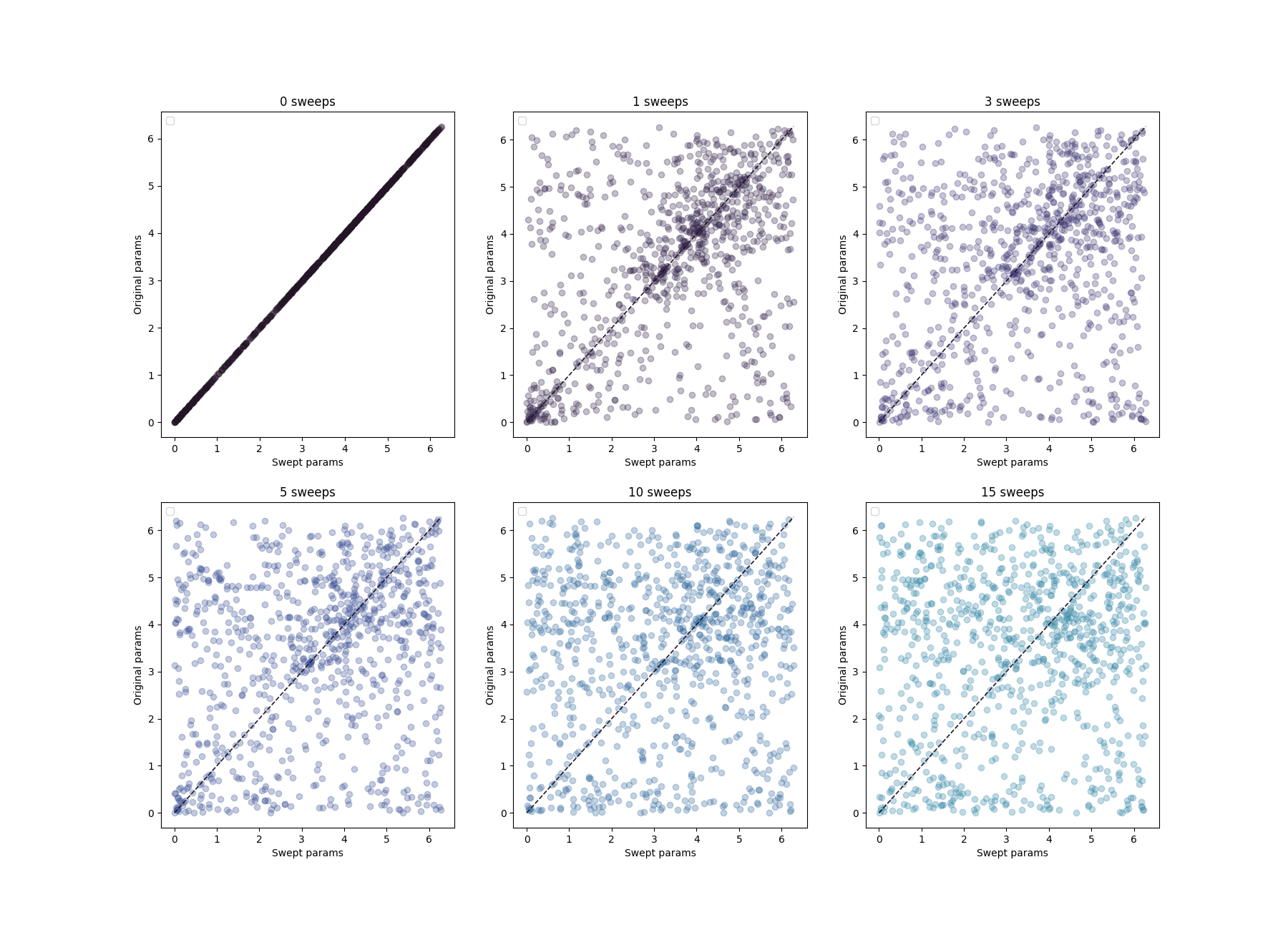}
  \caption{Angle sweeping. Taking a brickwork of two qubit gates and one qubit gates, we show the drift of angles from their original values after successive rounds of sweeping. The dashed indicates perfectly matching the original distribution.}
  \label{fig:sweep_angles}
\end{figure}

\textbf{Correlating gate connectivity}
Beyond the angles themselves, one may hope to unravel the circuit by looking at the gate connectivity in different parts of the circuit. The swaps and permutations make immediate cancellations less obvious, but one could try to reveal the permutations by identifying the relabeling. An attacker could potentially use the connectivity of the gates themselves, e.g. by looking at pairs of qubits that receive two-qubit interactions. Doing this identification would essentially be solving graph isomorphism, which is in NP, though not NP-complete. 
However, with the addition of masking, which replaces a subset of gates with a different set of gates, connectivity changes can be introduced, leading to an even more difficult version of the problem.

\textbf{Transpiler simplifications}
A primary sanity check of these techniques is to make transpilers, such as the one in Qiskit \cite{qiskit_paper}, struggle to identify and simplify the identity block. To illustrate the efficacy of the compiler, we perform a simple test in which an increasing number of swap transformations are enacted in a densely packed all-to-all random circuit and its inverse. Here, every layer contains $N/2$ two-qubit gates where neighbors are determined by taking a random permutation and pairing sequential qubits.  Indeed, before any swap transformations are introduced, the transpiler can find the full simplification, eliminating all two-qubit gates. 

While one may expect that the swaps have a gradual effect on the size of remaining gates which were not transpiled away, we find that actually the difficulty jumps relatively quickly. One can think of this swap interaction as being a bottleneck in the simplification of the full circuit. Thus, even a few bottlenecks may lead to a large fraction of gates surviving in the lightcone of those swaps. 

\begin{figure}[H]
  \centering

  {\includegraphics[width=.54\textwidth]{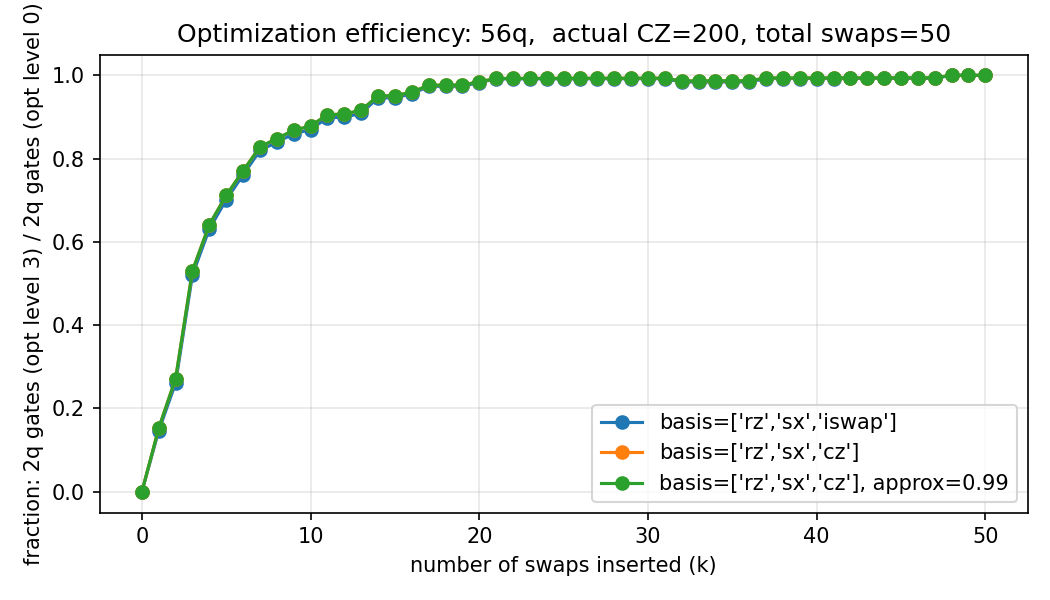}\label{fig:obfuscation_vs_rounds}}
  \caption{Transpiling away the identity. We show the difficulty of transpilers in simplifying these circuits. We start with an identity circuit comprised of $U$ and $U^\dagger$. The qiskit transpiler easily identifies it and reduces it. We then implement a series of swap transformations gradually and show the fraction of two qubits that survive the transpilation.   }  
\end{figure}

\textbf{Middle MPO Attack} To test the resiliency of our HQAP circuits, we employ a MPO compression technique we call the \textit{middle MPO attack}. To build intuition, consider how time evolution is applied to MPS. A gate gets applied as a tensor operation before performing compression back to the MPS form, possibly with higher bond dimension. Our MPO attack is based on a similar principle, except it can be in the middle of the circuit (with gates before and after), as opposed to an MPS with all gates in the future. 

This attack is particularly well suited for circuits which build up the identity from a sequence of random gates followed by the adjoint sequence. By starting with a rank 1 MPO object initialized to the identity in the middle of the circuit, joint evolution of gates on either side results in an MPO which is again representable with low rank since it remains the identity. Repeating this reveals that the whole circuit indeed has a much simpler representation and does not require simulation with exponentially large dimension. 

The combination of sweeping, masking, and applying permutations has thus far prevented this attack from successfully compressing the circuit. This shows that HQAP circuits may remain hard even when the boundary between where $T(U)$ ends and $U^\dagger$ begins is known. Though it is known that worst-case identity checking is QMA complete~\cite{doi:10.1142/S0219749905001067}, an open direction of importance is understanding how our version of identity obfuscation compares. In the meantime, we invite practictioners to try this and other strategies in our online portal. See Section~\ref{sec:open-peaked-challenge}.

\section{HQAP Circuits as a potential candidate for post-quantum encryption}
Throughout the work, we have studied the hardness of finding the peak when the input state is known and set to the all-0 bitstring. In this section, we discuss the possibility of using HQAP circuits as a post-quantum cryptographic protocol~\cite{chen2016report}. We first prove that, for generic quantum circuits, deciding whether a given input circuit is peaked is QCMA-complete. Note this decision problem of distinguishing peaked circuits differs from the task we consider in the previous sections, which is to produce the peaked bitstring given a circuit that is peaked. Therefore, unless BQP = QCMA, which is unlikely to be true under standard complexity assumptions, even a quantum computer cannot distinguish a peaked circuit from a non-peaked one. As a corollary, finding the input computational basis on which a circuit is peaked is hard even for a quantum computer. This provides a necessary condition that HQAP-based encryption can be secure: an efficient adversary without the trapdoor cannot even decide peakedness with non-negligible advantage, let alone recover the planted input string and obtaining the encrypted output string or weight.

On the other hand, the obfuscation protocol we build in this work permits efficient generation and would generate a non-negligible signal for a designated input-output string pair. Therefore, \emph{if} there exists an HQAP algorithm such that ``finding the input string on which the circuit is peaked'' is sufficiently hard on average, we can realize a post-quantum encryption scheme by treating the input state as a secret trapdoor: the encryption party publishes an obfuscated, peaked circuit for the target string, while the decryption party applies the trapdoor to recover it. In this case, correctness follows from peakedness, and security reduces to the hardness of learning the peak without the trapdoor.
% \iffalse{}
% We can use the fact that we can choose any bitstring to be the peak to encode (e.g. encrypt) messages into HQAP circuits.

% \subsection{Classical-safe Encryption}
% We can encrypt messages into HQAP circuits and that encryption will be classical-safe assuming our Heuritic Quantum Advantage claim stands. The high level protocol for such encryption will be:
% \begin{itemize}
%     \item Turn the message $M$ into a binary 01-bitstring $B_M$
%     \item Cut the message $B_M$ into series of length-$n$ 01-bitstrings $s_i$, where $n$ is the number of qubits of the quantum computer that will be able to solve this HQAP circuit.
%     \item Build HQAP circuits $C_i$ peaked at bitstring $s_i$
%     \item $C_i$ is now is the classical-safe encryption of M.
% \end{itemize}
% Only powerful enough quantum computers will be able to decrypt the message, e.g. find the peak bitstirng from the HQAP circuit. One drawback of such encryption technique will be the throughput - our current HQAP circuits would require hours to decrypt 56 bits. With different configuration changes, such that few shots is enough to detect the hidden peak, we can go from bytes/hour to speeds of KB/hour. 

% \subsection{Quantum-safe Encryption: \hayk{[to fill more]}}
% \fi{}
\subsection{Deciding whether a quantum circuit is peaked is QCMA-hard}

To this end, we cast the task of distinguishing a peaked circuit from a non-peaked one as the following decision problem:
\begin{definition}[Peakedness check on basis states (PCBS)]
\label{def:PCBS}
Let $U$ be an $n$-qubit unitary and fix thresholds $1 \ge \delta_{\rm yes} > \delta_{\rm no} \ge 0$ with
$\delta_{\rm yes}-\delta_{\rm no} \ge 1/\mathrm{poly}(N)$. The promise problem \textsc{PCBS} asks to decide whether
\[
\textsf{YES: }\exists\,y\in\{0,1\}^N \text{ with } |\langle y|U|y\rangle|\ge \delta_{\rm yes}
\quad\text{vs}\quad
\textsf{NO: }\forall\,y\in\{0,1\}^N,\ |\langle y|U|y\rangle|\le \delta_{\rm no}\,.
\]
\end{definition}
Next, we introduce a complexity class called QCMA. Intuitively, QCMA is the probabilistic version of NP (more formally, MA) with a quantum verifier and a classical witness.

\begin{definition}[QCMA]

Fix completeness--soundness thresholds $1 \ge c > s \ge 0$ with gap $c-s \ge 1/\mathrm{poly}(N)$.
A language $L\subseteq\{0,1\}^\star$ is in $\mathrm{QCMA}$ if there is a classical,
polynomial-time uniform generator that, on input $x$, outputs a polynomial-size quantum
verifier circuit $W_x$ acting on registers
\[
Y \text{ (witness, } N_x \text{ qubits)},\quad
A \text{ (ancillas, } M_x \text{ qubits initialized to } \ket{0^{M_x}}),\quad
B \text{ (one-qubit output)},
\]
such that
\[
\begin{aligned}
\textsf{YES: } & \exists\, y \in \{0,1\}^{N_x}\ \ 
\Pr\!\big[B{=}1 \ \text{after}\ W_x \ \text{on}\ \ket{y}_Y\ket{0^{M_x}}_A\ket{0}_B\big] \ \ge\ c,\\
\textsf{NO: } & \forall\, y \in \{0,1\}^{N_x}\ \ 
\Pr\!\big[B{=}1 \ \text{after}\ W_x \ \text{on}\ \ket{y}_Y\ket{0^{M_x}}_A\ket{0}_B\big] \ \le\ s.
\end{aligned}
\]
Equivalently, writing $\Pi_B := \ket{1}\!\bra{1}_B$ and
$\rho_y := \ket{y}\!\bra{y}_Y \otimes \ket{0^{M_x}}\!\bra{0^{M_x}}_A \otimes \ket{0}\!\bra{0}_B$,
the conditions are $\mathrm{tr}(W_x \rho_y W_x^\dagger \Pi_B) \ge c$ vs.\ $\le s$.
\end{definition}

\medskip

\begin{theorem}[PCBS is QCMA-complete]
\label{thm:PCBS-QCMA-complete}
\textsc{PCBS} is QCMA-complete under a polynomial-time reduction.
\end{theorem}
From the definition it is clear that PCBS has the right flavor for QCMA:
in QCMA the verifier's acceptance probability on some classical witness
separates YES and NO instances; in PCBS we ask whether some diagonal
entry of $U$ is large. First, it's easy to see the containment of PCBS $\in$ QCMA:
a QCMA verifier uses $y$ as the witness and estimates
$\Pr[\text{meas.\ }y\text{ after }U|y\rangle]=|\langle y|U|y\rangle|^2$ by sampling; since
$\delta_{\rm yes}^2-\delta_{\rm no}^2\ge 1/\mathrm{poly}(N)^2$, this is distinguishable in $\mathrm{poly}$ time.

Next, we show QCMA-hardness via a gadget that maps acceptance probabilities of any QCMA witness circuit $U$ to solving PCBS problem on an embedded larger circuit $Z$. This proof is inspired by~\cite{wocjan2003two}.

\begin{theorem}[]
\label{thm:flipped-T1}
Let $U$ be a circuit on registers $(Y,A,B)$ and let $P:=\Pi_{B=1}$.
For a basis witness $y$, define
\[
p(y):=\big\langle y,0^a,0\big|\,U^\dagger P U\,\big|y,0^a,0\big\rangle\in[0,1].
\]
Fix $0<\phi<\tfrac{\pi}{2}$ and the one-qubit rotation
$R(\phi)=\begin{pmatrix}\cos\phi&-\sin\phi\\ \sin\phi&\cos\phi\end{pmatrix}$ on a marker qubit $C$.
Define
\[
R_1:=\text{apply }R(\phi)\text{ to }C\text{ controlled on }(A{=}0^a),\qquad
R_2:=\text{apply }R(-\phi)\text{ to }C\text{ controlled on }(B{=}1),
\]
and introduce a fresh guard qubit $D$ (initialized to $\ket0$) with the unitary
\[
G\;:=\;X_D\otimes\bigl(I_{YABC}-\ket{0^{a}0}\!\bra{0^{a}0}_{AB}\bigr)\;+\;I_D\otimes \ket{0^{a}0}\!\bra{0^{a}0}_{AB}.
\]
Next, we set
\[
Z\ :=\ U^\dagger\,R_2\,U\,R_1\,G.
\]
For the canonical basis state $|z\rangle:=|0\rangle_C|0\rangle_D\otimes|y,0^a,0\rangle_{YAB}$ we have
\begin{equation}
\label{eq:interval}
\bigl|\langle z|Z|z\rangle\bigr|^2\ =\ \bigl(p(y) + (1-p(y))\cos\phi\bigr)^{\!2}.
\end{equation}
Moreover, for any basis state $|z'\rangle$ with $(A,B)\neq(0^a,0)$, that is: when the ancillas are not in the canonical input staes
\begin{equation}
\label{eq:noncanon}
\bigl\langle z'\big|\,Z\,\big|z'\bigr\rangle\ =\ 0\qquad\text{(since $G$ flips $D$ to an orthogonal basis state).}
\end{equation}

\paragraph{YES bound.}
If there exists $y$ with $p(y)\ge 1-\varepsilon$, then for the corresponding canonical $|z\rangle$,
\begin{equation}
\label{eq:yes-explicit}
\bigl|\langle z|Z|z\rangle\bigr|^2 \ \ge\
\Bigl((1-\varepsilon)+\varepsilon\cos\phi\Bigr)^{\!2}
\ =\ \Bigl(1-\varepsilon(1-\cos\phi)\Bigr)^{\!2}.
\end{equation}

\paragraph{NO bound.}
If for all $y$ we have $p(y)\le \varepsilon$, then for every canonical $|z\rangle$,
\begin{equation}
\label{eq:no-explicit}
\bigl|\langle z|Z|z\rangle\bigr|^2 \ \le\
\Bigl(\varepsilon+(1-\varepsilon)\cos\phi\Bigr)^{\!2}
\ \le\ \bigl(\cos\phi+\varepsilon\bigr)^{\!2}.
\end{equation}
\end{theorem}

\begin{proof}[Proof]
\emph{Step 0 (apply $G$).}
Recall $Z = U^\dagger R_2 U R_1 G$, so $G$ acts first.
If $(A,B)=(0^a,0)$, $G$ is the identity and the analysis below applies unchanged.
If $(A,B)\neq(0^a,0)$, i.e., in the non-canonical input case, then $G$ flips $D$; since the rest of the circuit acts trivially on $D$,
$\langle z'|Z|z'\rangle=0$, proving \eqref{eq:noncanon}.

In the canonical case, we expand
\(
U\,|y,0^a,0\rangle
 = c_1\,|1\rangle_B|\psi_1\rangle + c_0\,|0\rangle_B|\psi_0\rangle,
\)
with $|c_1|^2=p(y)$ and $|c_0|^2=1-p(y)$.

\emph{Step 1 (apply $R_1$).}
Since $A=0^a$ on the canonical input, $R_1$ rotates the marker:
\(
|0\rangle_C \mapsto |c_0\rangle_C:=R(\phi)|0\rangle
=\cos\phi\,|0\rangle+\sin\phi\,|1\rangle.
\)

\emph{Step 2 (apply $U$).} We get
\[
|c_0\rangle_C\otimes U|y,0^a,0\rangle
= c_1\,|c_0\rangle_C\otimes|1\rangle_B|\psi_1\rangle
 + c_0\,|c_0\rangle_C\otimes|0\rangle_B|\psi_0\rangle.
\]

\emph{Step 3 (apply $R_2$ controlled on $B=1$).}
On the $B{=}1$ branch the marker sees $R(-\phi)$, so $R(-\phi)R(\phi)|0\rangle=|0\rangle$; on the $B{=}0$ branch nothing happens.
Hence the joint state becomes
\[
c_1\,|0\rangle_C|1\rangle_B|\psi_1\rangle \ +\ c_0\,|c_0\rangle_C|0\rangle_B|\psi_0\rangle.
\]

\emph{Step 4 (apply $U^\dagger$ and overlap).}
Applying $U^\dagger$ on $(YAB)$ and projecting onto $\langle 0|_C\otimes\langle y,0^a,0|$ gives
\[
\langle z|U^\dagger R_2 U R_1|z\rangle
= |c_1|^2 + \cos\phi\,|c_0|^2
= p(y) + (1-p(y))\cos\phi,
\]
which proves \eqref{eq:interval}. Notice that the quantity is real and nonnegative for $\phi\in(0,\tfrac\pi2)$.
The bounds \eqref{eq:yes-explicit}-\eqref{eq:no-explicit} follow from the monotonicity of $f(p):=p+(1-p)\cos\phi$.
\end{proof}

Given a QCMA instance $x$ with verifier $W_x$ and thresholds $c>s$, build $Z$ from $U:=W_x$ as in Theorem~\ref{thm:flipped-T1} using any fixed constant $\phi\in(0,\tfrac\pi2)$ (e.g.\ $\phi=\pi/3$).
Set PCBS thresholds
\[
\delta_{\rm yes}:=f(c)=c+(1-c)\cos\phi,\qquad
\delta_{\rm no}:=f(s)=s+(1-s)\cos\phi.
\]
Then $\delta_{\rm yes}-\delta_{\rm no}=(1-\cos\phi)(c-s)\ge 1/\mathrm{poly}(N)$, and:
\[
x\in L \Longleftrightarrow \exists\,y:\ |\langle z|Z|z\rangle| \ge \delta_{\rm yes}
\quad\text{vs}\quad
x\notin L \Longleftrightarrow \forall\,y:\ |\langle z|Z|z\rangle|\le \delta_{\rm no},
\]
establishing QCMA-hardness. Together with our previous containment argument, this proves Theorem~\ref{thm:PCBS-QCMA-complete}. An immediate corollary of this theorem is:
\begin{corollary}
    There exists no quantum polynomial-time algorithm that, given a generic $\text{poly}(N)$-sized quantum circuit $U$
(described classically), outputs the peaked input string with non-negligible advantage than random guess, unless $\mathrm{BQP}=\mathrm{QCMA}$. 
\end{corollary}
This can be proved by contradiction: assume that a quantum polynomial-time (QPT) algorithm can return the secret peaked input, then running the quantum circuit would generate the desired peakedness, resolving PCBS.
\subsection{Symmetric encryption with HQAP circuits}
The QCMA-hardness in the previous section suggests that finding the peaked input--output pair of a
peaked circuit is hard even for quantum algorithms. This motivates using HQAP-generated circuits
as ciphertexts: in our obfuscation protocols, a quantum adversary is not expected to efficiently
identify the designated peaked input or output. Intuitively, hiding a single input--output pair
among exponentially many basis labels makes it difficult to locate without the trapdoor. We conjecture:

\begin{conjecture}[Peak-Search Hardness (PSH)]
    There exist an efficient HQAP protocol which, given $(x^\star,y^\star)\in\{0,1\}^N\times\{0,1\}^N$ and fresh coins $\rho$,
the procedure $\mathsf{HQAP}(x^\star,y^\star;\rho)$ outputs a circuit $E$ such that
measuring $E\ket{x^\star}$ in the computational basis returns $y^\star$ with probability at least
$\delta\in(0,1)$. For all $x\neq x^\star$, the output distribution of $E\ket{x}$ is not
comparably concentrated on $y^\star$. 

Moreover, the security of this protocol is guaranteed in the adversarial setting: give the adversary the classical
description $\mathrm{desc}(E)$. Then no QPT outputs $x^\star$ (or the pair
$(x^\star,y^\star)$) with probability better than $2^{-n}+\text{negl}(n)$.
\end{conjecture}

Then with a slight modification to that HQAP circuit protocol can also be used for encrypting messages with quantum-safe property. Here is a sketch of the protocol: Let $k\in\{0,1\}^d$ be a pre-shared key and $F_k$ a pseudo random function (PRF) to derive per-block designated inputs.
To encrypt $M$, parse it into $n$-bit blocks $s_1,\dots,s_t$. For each block $i$:
\begin{enumerate}
  \item Publish a nonce $\mathsf{ctr}_i$ and set $x_i^\star := F_k(\mathsf{ctr}_i)$.
  \item Set $y_i^\star := s_i$ (optionally, mask as $s_i\oplus G_k(\mathsf{ctr}_i)$ with an independent PRF $G$).
  \item Sample $\rho_i$ and compute $E_i \leftarrow \mathsf{HQAP}(x_i^\star,y_i^\star;\rho_i)$.
  \item Output ciphertext $C_i := \big(\mathrm{desc}(E_i),\,\mathsf{ctr}_i\,\big)$.
\end{enumerate}
The receiver decrypts by computing $x_i^\star=F_k(\mathsf{ctr}_i)$, running $E_i$ on $\ket{x_i^\star}$ for
$N_s=O(\delta^{-2})$ shots, taking the empirical mode $\hat y_i$, and outputting
$\hat s_i:=\hat y_i$. As a check, the receiver accepts only if the measured peak weight
$\widehat{p}^{\,\text{peak}}=\frac{1}{N_s}\sum_{j=1}^N\mathbf{1}\{y_{i,j}=\hat y_i\}$ exceeds a
threshold $\delta_{\min}$.

It's not hard to see that the protocol remains secure under any QPT adversary
against one-wayness under chosen-plaintext attacks (OW-CPA): Assume PRF security for $F$ and PSH for the HQAP distribution.  Given a ciphertext block $C_i=(\mathrm{desc}(E_i),\mathsf{ctr}_i)$, recovering $s_i$ requires either (a) computing
$x_i^\star=F_k(\mathsf{ctr}_i)$ (which breaks the PRF), or (b) finding $x_i^\star$ by probing $E_i$
on inputs and detecting the heavy output (which breaks PSH), then reading off $y_i^\star$ and
removing the mask via $G_k$ (which again requires the PRF). Thus, under PRF security and PSH, the
scheme achieves OW-CPA.

On the receiver side, checking if the output peakedness have reached a designated threshold provides a practical sanity layer against malformed or dishonest ciphertexts. Beyond this, this HQAP protocol offers three unique properties other PQC schemes may lack: (i) validity here is semantic, an inherent property of how the ciphertext was generated and, crucially, receiver-local; (ii) the same machinery naturally doubles as a proof-of-quantumness-flavored check; and (iii) unlike standardized post-quantum cryptography (e.g., ML-KEM/ML-DSA/SLH-DSA)~\cite{chen2016report}, which is deliberately classical and deployable on today's networks, our protocol is inherently quantum because decryption or validation requires executing the public circuit on the designated input. Of course, these properties hinge on an average-case PSH conjecture holding for our instance generation: this is precisely why establishing rigorously or empirically testing on the hardness properties of HQAP circuits remain a crucial open problem of broad applications.

\iffalse{}
\emph{if} there exist a HQAP protocol st this is still effi
With a slight modification to our HQAP circuit protocol \ref{sec:peaked_protocol} we can also use it for encrypting messages with quantum-safe property. The key idea is we can cut and keep a random chunk of the beginning of our HQAP circuit and use that as the private key of the encoding. The full high level protocol then will be:
\begin{itemize}
    \item Turn the message $M$ into a binary 01-bitstring $B_M$
    \item Cut the message $B_M$ into series of length-$n$ 01-bitstrings $s_i$, where $n$ is the number of qubits of the quantum computer that will be able to solve this HQAP circuit.
    \item Fix a chunk of $T(R)$ as private key $P_k$, and build HQAP circuits $C_i$ = $P_k$ + $E_i$ peaked at bitstring $s_i$ with the same $P_k$ as the starting sub-circuit.
    \item $E_i$ is now is the quantum-safe encryption of M.
\end{itemize}

Now even quantum computers should not be able to decode the circuits $E_i$ into $M$, but with the private key $P_k$ they can recover the HQAP circuits $C_i$, run the circuits, find the peaks and decode the message M.
The drawbacks, state preparation and throughput discussion are similar to the classical-safe protocol. 
\fi{}

\section{Discussion}
In this work, we have presented a protocol for constructing peaked quantum circuits that is both scalable and verifiable. Our construction builds on the original peaked circuit framework by introducing pre-processing steps and architectural refinements that allow these circuits to scale well beyond previous instantiations, in both qubit number and circuit depth. Importantly, we have designed specific circuits that stretch the computational limits of existing quantum processors while fitting within the error budget. 
At the same time, we have implemented a suite of classical simulation attack strategies, ranging from matrix product states to tensor network contraction methods to Pauli path simulators, that represent the current state of the art in classical circuit simulation. Many of these methods have been used in past work to challenge or reinterpret quantum advantage claims, including those involving IBM's utility-scale benchmarking efforts. 

These results offer a timely contribution to the evolving landscape of quantum advantage. It is now broadly accepted that certain quantum processors can outperform classical computers on well-posed computational tasks, at least in the absence of strong classical shortcuts. However, in the case of RCS, these demonstrations are often based on output distributions that are difficult to interpret or verify at scale. Our work addresses this gap by providing a construction in which the quantum output is both easy to produce and easy to verify, yet remains classically intractable to simulate by all known methods. This puts us in what might be described as the \textit{heuristic} phase of quantum advantage: a period in which empirical results, rather than unconditional complexity-theoretic proofs, serve as the main form of evidence.

In this sense, our peaked circuits occupy a role not unlike that of RSA in classical cryptography. Just as the hardness of factoring large integers is widely believed, despite the lack of a formal proof, because of decades of failed attempts to break it, the intractability of simulating peaked circuits may come to be seen as an empirical fact. 

On the other hand, there may be continual progress on the theoretical side, bridging the gap between what is provably hard and our constructions. In this work, we build on the existing literature on the hardness of peaked circuits. In general, a peaked circuit may lead to concentration on a particular output bitstring from an input state other than the all zeros case. When the input bitstring is unknown, determining whether a given circuit is peaked is QCMA-hard. Establishing a concrete connection between these complexity results and our existing construction would greatly add to the strength of claims of quantum advantage.

We do not claim that classical simulation of our peaked circuits is impossible. Indeed, further analysis may reveal weaknesses in specific instances, or new classical algorithms may improve upon current simulation strategies. But in the absence of such breakthroughs, we view these circuits as offering a valuable benchmark for the quantum community. They provide a target that is both meaningful and falsifiable, anyone who can simulate them classically is encouraged to do so (see Section~\ref{sec:open-peaked-challenge}). In this way, our work is as much a challenge as it is a result. By continuing to scale these constructions and openly testing their classical boundaries, we hope to sharpen the contours of what practical quantum advantage looks like in the NISQ era.

\section{Open Peaked Circuits challenge}\label{sec:open-peaked-challenge}
We have observed a big gap between quantum and all known classical tenchiques, however we do not exclude the possibility that a smarter classical technique exists.
Hence we are announcing a peaked circuits challenge, where we will be releasing increasingly difficult peaked circuits to cement the runtime gap between quantum and classical solutions. These circuits can be found on 
\href{https://app.bluequbit.io/hackathons/oEOtLSSrPSVH60Ah}{BlueQubit Peak Portal} where anyone is welcome to submit their classical solutions as well. This is similar to the RSA challenge announced in the 1990s and expanded in the 2000s~\cite{Kaliski2005}, where the challenge was to factor a large product of primes. The hardness of the challenge was never formally proved, but it was heuristically agreed upon based on many tries from the community. 

\section*{Acknowledgements}
We thank the entire BlueQubit team for contributions to the BlueQubit SDK and, in particular, to the Pauli Path and MPS simulators used for the classical simulations in this work. We are grateful to the Quantinuum team -- particularly its Startup Partner Program --- for sustained support and for facilitating access to the System Model H2 trapped-ion processor. We also thank Scott Aaronson and Dima Abanin for numerous valuable discussions and feedback.

% \bibliographystyle{alpha}
% \bibliographystyle{unsrt}

% \bibliography{refs_gibbs}
\printbibliography

@article{Arute2019QuantumSupremacy,
  title     = {Quantum supremacy using a programmable superconducting processor},
  author    = {Arute, F. and Arya, K. and Babbush, R. and others},
  journal   = {Nature},
  volume    = {574},
  pages     = {505--510},
  year      = {2019},
  doi       = {10.1038/s41586-019-1666-5},
  url       = {https://doi.org/10.1038/s41586-019-1666-5}
}

@article{wocjan2003two,
  title={Two QCMA-complete problems},
  author={Wocjan, Pawel and Janzing, Dominik and Beth, Thomas},
  journal={arXiv preprint quant-ph/0305090},
  year={2003}
}

@article{zhang2025classical,
  title={Classical simulability of quantum circuits with shallow magic depth},
  author={Zhang, Yifan and Zhang, Yuxuan},
  journal={PRX Quantum},
  volume={6},
  number={1},
  pages={010337},
  year={2025},
  publisher={APS}
}

@article{Morvan2024PhaseTransitions,
  title     = {Phase transitions in random circuit sampling},
  author    = {Morvan, A. and Villalonga, B. and Mi, X. and others},
  journal   = {Nature},
  volume    = {634},
  pages     = {328--333},
  year      = {2024},
  doi       = {10.1038/s41586-024-07998-6},
  url       = {https://doi.org/10.1038/s41586-024-07998-6}
}

@article{PhysRevX.15.021052,
  title = {Computational Power of Random Quantum Circuits in Arbitrary Geometries},
  author = {DeCross, M. and Haghshenas, R. and Liu, M. and Rinaldi, E. and Gray, J. and Alexeev, Y. and Baldwin, C. H. and Bartolotta, J. P. and Bohn, M. and Chertkov, E. and Cline, J. and Colina, J. and DelVento, D. and Dreiling, J. M. and Foltz, C. and Gaebler, J. P. and Gatterman, T. M. and Gilbreth, C. N. and Giles, J. and Gresh, D. and Hall, A. and Hankin, A. and Hansen, A. and Hewitt, N. and Hoffman, I. and Holliman, C. and Hutson, R. B. and Jacobs, T. and Johansen, J. and Lee, P. J. and Lehman, E. and Lucchetti, D. and Lykov, D. and Madjarov, I. S. and Mathewson, B. and Mayer, K. and Mills, M. and Niroula, P. and Pino, J. M. and Roman, C. and Schecter, M. and Siegfried, P. E. and Tiemann, B. G. and Volin, C. and Walker, J. and Shaydulin, R. and Pistoia, M. and Moses, S. A. and Hayes, D. and Neyenhuis, B. and Stutz, R. P. and Foss-Feig, M.},
  journal = {Phys. Rev. X},
  volume = {15},
  issue = {2},
  pages = {021052},
  numpages = {39},
  year = {2025},
  month = {May},
  publisher = {American Physical Society},
  doi = {10.1103/PhysRevX.15.021052},
  url = {https://link.aps.org/doi/10.1103/PhysRevX.15.021052}
}

@software{quantum_ai_team_and_collaborators_2020_4023103,
  author    = {Quantum AI team and collaborators},
  title     = {qsim},
  year      = {2020},
  month     = {9},
  publisher = {Zenodo},
  doi       = {10.5281/zenodo.4023103},
  url       = {https://doi.org/10.5281/zenodo.4023103}
}

@article{Bayraktar2023cuQuantum,
  title   = {cuQuantum SDK: A High-Performance Library for Accelerating Quantum Science},
  author  = {Bayraktar, Harun and Charara, Ali and Clark, David and Cohen, Saul and Costa, Timothy and Fang, Yao-Lung L. and Gao, Yang and Guan, Jack and Gunnels, John and Haidar, Azzam and Hehn, Andreas and Hohnerbach, Markus and Jones, Matthew and Lubowe, Tom and Lyakh, Dmitry and Morino, Shinya and Springer, Paul and Stanwyck, Sam and Terentyev, Igor and Varadhan, Satya and Wong, Jonathan and Yamaguchi, Takuma},
  journal = {arXiv preprint arXiv:2308.01999},
  year    = {2023},
  doi     = {10.48550/arXiv.2308.01999}
}

@article{Gray2018quimb,
  title   = {quimb: A python package for quantum information and many-body calculations},
  author  = {Gray, Johnnie},
  journal = {Journal of Open Source Software},
  volume  = {3},
  number  = {29},
  pages   = {819},
  year    = {2018},
  doi     = {10.21105/joss.00819}
}

@article{Rall2019PauliPropagation,
  title   = {Simulation of qubit quantum circuits via Pauli propagation},
  author  = {Rall, Patrick and Liang, Daniel and Cook, Jeremy and Kretschmer, William},
  journal = {Phys. Rev. A},
  volume  = {99},
  pages   = {062337},
  year    = {2019},
  doi     = {10.1103/PhysRevA.99.062337}
}

@article{Rudolph:2025gyq,
    author = {Rudolph, Manuel S. and Jones, Tyson and Teng, Yanting and Angrisani, Armando and Holmes, Zo{\"e}},
    title = "{Pauli Propagation: A Computational Framework for Simulating Quantum Systems}",
    eprint = "2505.21606",
    archivePrefix = "arXiv",
    primaryClass = "quant-ph",
    month = "5",
    year = "2025"
}

@article{Gharibyan:2025ldn,
    author = "Gharibyan, Hrant and Hariprakash, Siddharth and Mullath, Mohammed Zuhair and Su, Vincent P.",
    title = "{A Practical Guide to using Pauli Path Simulators for Utility-Scale Quantum Experiments}",
    eprint = "2507.10771",
    archivePrefix = "arXiv",
    primaryClass = "quant-ph",
    month = "7",
    year = "2025"
}

@article{Schuster:2024jds,
    author = "Schuster, Thomas and Yin, Chao and Gao, Xun and Yao, Norman Y.",
    title = "{A polynomial-time classical algorithm for noisy quantum circuits}",
    eprint = "2407.12768",
    archivePrefix = "arXiv",
    primaryClass = "quant-ph",
    month = "7",
    year = "2024"
}

@inproceedings{Aharonov2023NoisyRCS,
  title     = {A Polynomial-Time Classical Algorithm for Noisy Random Circuit Sampling},
  author    = {Aharonov, Dorit and Gao, Xun and Landau, Zeph and Liu, Yunchao and Vazirani, Umesh},
  booktitle = {Proceedings of the 55th Annual ACM Symposium on Theory of Computing (STOC '23)},
  year      = {2023},
  doi       = {10.1145/3564246.3585234},
  note      = {See also arXiv:2211.03999 (2022)}
}

@article{
doi:10.1126/sciadv.adk4321,
author = {Tomislav Begušić  and Johnnie Gray  and Garnet Kin-Lic Chan },
title = {Fast and converged classical simulations of evidence for the utility of quantum computing before fault tolerance},
journal = {Science Advances},
volume = {10},
number = {3},
pages = {eadk4321},
year = {2024},
doi = {10.1126/sciadv.adk4321},
URL = {https://www.science.org/doi/abs/10.1126/sciadv.adk4321},
eprint = {https://www.science.org/doi/pdf/10.1126/sciadv.adk4321},
}

@article{Begusic:2023owa,
    author = "Begu{\v{s}}i{\'c}, Tomislav and Hejazi, Kasra and Kin-Lic Chan, Garnet",
    title = "{Simulating quantum circuit expectation values by Clifford perturbation theory}",
    eprint = "2306.04797",
    archivePrefix = "arXiv",
    primaryClass = "quant-ph",
    doi = "10.1063/5.0269149",
    journal = "J. Chem. Phys.",
    volume = "162",
    number = "15",
    pages = "154110",
    year = "2025"
}

@article{PRXQuantum.6.020302,
  title = {Real-Time Operator Evolution in Two and Three Dimensions via Sparse Pauli Dynamics},
    author = "Begu{\v{s}}i{\'c}, Tomislav and Kin-Lic Chan, Garnet",
  journal = {PRX Quantum},
  volume = {6},
  issue = {2},
  pages = {020302},
  numpages = {13},
  year = {2025},
  month = {Apr},
  publisher = {American Physical Society},
  doi = {10.1103/PRXQuantum.6.020302},
  url = {https://link.aps.org/doi/10.1103/PRXQuantum.6.020302}
}

@misc{aaronson_2024_peaked,
      title={On verifiable quantum advantage with peaked circuit sampling}, 
      author={Scott Aaronson and Yuxuan Zhang},
      year={2024},
      eprint={2404.14493},
      archivePrefix={arXiv},
      primaryClass={quant-ph},
      url={https://arxiv.org/abs/2404.14493}, 
}

@article{Gray:2020cah,
    author = "Gray, Johnnie and Kourtis, Stefanos",
    title = "{Hyper-optimized tensor network contraction}",
    eprint = "2002.01935",
    archivePrefix = "arXiv",
    primaryClass = "quant-ph",
    doi = "10.22331/q-2021-03-15-410",
    journal = "Quantum",
    volume = "5",
    pages = "410",
    year = "2021"
}

@article{Biamonte_2017,
   title={Quantum machine learning},
   volume={549},
   ISSN={1476-4687},
   url={http://dx.doi.org/10.1038/nature23474},
   DOI={10.1038/nature23474},
   number={7671},
   journal={Nature},
   publisher={Springer Science and Business Media LLC},
   author={Biamonte, Jacob and Wittek, Peter and Pancotti, Nicola and Rebentrost, Patrick and Wiebe, Nathan and Lloyd, Seth},
   year={2017},
   month=sep, pages={195–202} }

@article{McClean_2018-barren,
   title={Barren plateaus in quantum neural network training landscapes},
   volume={9},
   ISSN={2041-1723},
   url={http://dx.doi.org/10.1038/s41467-018-07090-4},
   DOI={10.1038/s41467-018-07090-4},
   number={1},
   journal={Nature Communications},
   publisher={Springer Science and Business Media LLC},
   author={McClean, Jarrod R. and Boixo, Sergio and Smelyanskiy, Vadim N. and Babbush, Ryan and Neven, Hartmut},
   year={2018},
   month=nov }

@article{ji2009non,
  title={Non-identity check remains QMA-complete for short circuits},
  author={Ji, Zhengfeng and Wu, Xiaodi},
  journal={arXiv preprint arXiv:0906.5416},
  year={2009}
}

@article{dam2023survey,
  title={A survey of post-quantum cryptography: Start of a new race},
  author={Dam, Duc-Thuan and Tran, Thai-Ha and Hoang, Van-Phuc and Pham, Cong-Kha and Hoang, Trong-Thuc},
  journal={Cryptography},
  volume={7},
  number={3},
  pages={40},
  year={2023},
  publisher={MDPI}
}

@article{wu2021strong,
  title={Strong quantum computational advantage using a superconducting quantum processor},
  author={Wu, Yulin and Bao, Wan-Su and Cao, Sirui and Chen, Fusheng and Chen, Ming-Cheng and Chen, Xiawei and Chung, Tung-Hsun and Deng, Hui and Du, Yajie and Fan, Daojin and others},
  journal={Physical review letters},
  volume={127},
  number={18},
  pages={180501},
  year={2021},
  publisher={APS}
}

@inproceedings{kumar2020post,
  title={Post quantum cryptography (pqc)-an overview},
  author={Kumar, Manoj and Pattnaik, Pratap},
  booktitle={2020 IEEE High Performance Extreme Computing Conference (HPEC)},
  pages={1--9},
  year={2020},
  organization={IEEE}
}

@article{joseph2022transitioning,
  title={Transitioning organizations to post-quantum cryptography},
  author={Joseph, David and Misoczki, Rafael and Manzano, Marc and Tricot, Joe and Pinuaga, Fernando Dominguez and Lacombe, Olivier and Leichenauer, Stefan and Hidary, Jack and Venables, Phil and Hansen, Royal},
  journal={Nature},
  volume={605},
  number={7909},
  pages={237--243},
  year={2022},
  publisher={Nature Publishing Group UK London}
}

@article{zhu2022quantum,
  title={Quantum computational advantage via 60-qubit 24-cycle random circuit sampling},
  author={Zhu, Qingling and Cao, Sirui and Chen, Fusheng and Chen, Ming-Cheng and Chen, Xiawei and Chung, Tung-Hsun and Deng, Hui and Du, Yajie and Fan, Daojin and Gong, Ming and others},
  journal={Science bulletin},
  volume={67},
  number={3},
  pages={240--245},
  year={2022},
  publisher={Elsevier}
}

@article{Moses_2023,
   title={A Race-Track Trapped-Ion Quantum Processor},
   volume={13},
   ISSN={2160-3308},
   url={http://dx.doi.org/10.1103/PhysRevX.13.041052},
   DOI={10.1103/physrevx.13.041052},
   number={4},
   journal={Physical Review X},
   publisher={American Physical Society (APS)},
   author={Moses, S. A. and Baldwin, C. H. and Allman, M. S. and Ancona, R. and Ascarrunz, L. and Barnes, C. and Bartolotta, J. and Bjork, B. and Blanchard, P. and Bohn, M. and Bohnet, J. G. and Brown, N. C. and Burdick, N. Q. and Burton, W. C. and Campbell, S. L. and Campora, J. P. and Carron, C. and Chambers, J. and Chan, J. W. and Chen, Y. H. and Chernoguzov, A. and Chertkov, E. and Colina, J. and Curtis, J. P. and Daniel, R. and DeCross, M. and Deen, D. and Delaney, C. and Dreiling, J. M. and Ertsgaard, C. T. and Esposito, J. and Estey, B. and Fabrikant, M. and Figgatt, C. and Foltz, C. and Foss-Feig, M. and Francois, D. and Gaebler, J. P. and Gatterman, T. M. and Gilbreth, C. N. and Giles, J. and Glynn, E. and Hall, A. and Hankin, A. M. and Hansen, A. and Hayes, D. and Higashi, B. and Hoffman, I. M. and Horning, B. and Hout, J. J. and Jacobs, R. and Johansen, J. and Jones, L. and Karcz, J. and Klein, T. and Lauria, P. and Lee, P. and Liefer, D. and Lu, S. T. and Lucchetti, D. and Lytle, C. and Malm, A. and Matheny, M. and Mathewson, B. and Mayer, K. and Miller, D. B. and Mills, M. and Neyenhuis, B. and Nugent, L. and Olson, S. and Parks, J. and Price, G. N. and Price, Z. and Pugh, M. and Ransford, A. and Reed, A. P. and Roman, C. and Rowe, M. and Ryan-Anderson, C. and Sanders, S. and Sedlacek, J. and Shevchuk, P. and Siegfried, P. and Skripka, T. and Spaun, B. and Sprenkle, R. T. and Stutz, R. P. and Swallows, M. and Tobey, R. I. and Tran, A. and Tran, T. and Vogt, E. and Volin, C. and Walker, J. and Zolot, A. M. and Pino, J. M.},
   year={2023},
   month=dec }

@inproceedings{bravyi2024classical,
  title={Classical simulation of peaked shallow quantum circuits},
  author={Bravyi, Sergey and Gosset, David and Liu, Yinchen},
  booktitle={Proceedings of the 56th Annual ACM Symposium on Theory of Computing},
  pages={561--572},
  year={2024}
}

@article{SCHOLLWOCK201196,
title = {The density-matrix renormalization group in the age of matrix product states},
journal = {Annals of Physics},
volume = {326},
number = {1},
pages = {96-192},
year = {2011},
note = {January 2011 Special Issue},
issn = {0003-4916},
doi = {https://doi.org/10.1016/j.aop.2010.09.012},
url = {https://www.sciencedirect.com/science/article/pii/S0003491610001752},
author = {Ulrich Schollwock}
}

@article{PhysRevResearch.3.023073,
  title = {Tensor networks contraction and the belief propagation algorithm},
  author = {Alkabetz, R. and Arad, I.},
  journal = {Phys. Rev. Res.},
  volume = {3},
  issue = {2},
  pages = {023073},
  numpages = {12},
  year = {2021},
  month = {Apr},
  publisher = {American Physical Society},
  doi = {10.1103/PhysRevResearch.3.023073},
  url = {https://link.aps.org/doi/10.1103/PhysRevResearch.3.023073}
}

@article{ORUS2014117,
title = {A practical introduction to tensor networks: Matrix product states and projected entangled pair states},
journal = {Annals of Physics},
volume = {349},
pages = {117-158},
year = {2014},
issn = {0003-4916},
doi = {https://doi.org/10.1016/j.aop.2014.06.013},
url = {https://www.sciencedirect.com/science/article/pii/S0003491614001596},
author = {Rom{\'a}n Or{\'u}s},
}

@Article{10.21468/SciPostPhys.15.6.222,
	title={{Gauging tensor networks with belief propagation}},
	author={Joseph Tindall and Matt Fishman},
	journal={SciPost Phys.},
	volume={15},
	pages={222},
	year={2023},
	publisher={SciPost},
	doi={10.21468/SciPostPhys.15.6.222},
	url={https://scipost.org/10.21468/SciPostPhys.15.6.222},
}

@misc{verstraete2004renormalizationalgorithmsquantummanybody,
      title={Renormalization algorithms for Quantum-Many Body Systems in two and higher dimensions}, 
      author={F. Verstraete and J. I. Cirac},
      year={2004},
      eprint={cond-mat/0407066},
      archivePrefix={arXiv},
      primaryClass={cond-mat.str-el},
      url={https://arxiv.org/abs/cond-mat/0407066}, 
}

@article{RevModPhys.93.045003,
  title = {Matrix product states and projected entangled pair states: Concepts, symmetries, theorems},
  author = {Cirac, J. Ignacio and P\'erez-Garc\'{\i}a, David and Schuch, Norbert and Verstraete, Frank},
  journal = {Rev. Mod. Phys.},
  volume = {93},
  issue = {4},
  pages = {045003},
  numpages = {65},
  year = {2021},
  month = {Dec},
  publisher = {American Physical Society},
  doi = {10.1103/RevModPhys.93.045003},
  url = {https://link.aps.org/doi/10.1103/RevModPhys.93.045003}
}

@article{PhysRevLett.91.147902,
  title = {Efficient Classical Simulation of Slightly Entangled Quantum Computations},
  author = {Vidal, Guifr\'e},
  journal = {Phys. Rev. Lett.},
  volume = {91},
  issue = {14},
  pages = {147902},
  numpages = {4},
  year = {2003},
  month = {Oct},
  publisher = {American Physical Society},
  doi = {10.1103/PhysRevLett.91.147902},
  url = {https://link.aps.org/doi/10.1103/PhysRevLett.91.147902}
}

@article{PhysRevLett.93.040502,
  title = {Efficient Simulation of One-Dimensional Quantum Many-Body Systems},
  author = {Vidal, Guifr\'e},
  journal = {Phys. Rev. Lett.},
  volume = {93},
  issue = {4},
  pages = {040502},
  numpages = {4},
  year = {2004},
  month = {Jul},
  publisher = {American Physical Society},
  doi = {10.1103/PhysRevLett.93.040502},
  url = {https://link.aps.org/doi/10.1103/PhysRevLett.93.040502}
}

@article{PhysRevLett.98.070201,
  title = {Classical Simulation of Infinite-Size Quantum Lattice Systems in One Spatial Dimension},
  author = {Vidal, G.},
  journal = {Phys. Rev. Lett.},
  volume = {98},
  issue = {7},
  pages = {070201},
  numpages = {4},
  year = {2007},
  month = {Feb},
  publisher = {American Physical Society},
  doi = {10.1103/PhysRevLett.98.070201},
  url = {https://link.aps.org/doi/10.1103/PhysRevLett.98.070201}
}

@article{PhysRevB.78.155117,
  title = {Infinite time-evolving block decimation algorithm beyond unitary evolution},
  author = {Or\'us, R. and Vidal, G.},
  journal = {Phys. Rev. B},
  volume = {78},
  issue = {15},
  pages = {155117},
  numpages = {11},
  year = {2008},
  month = {Oct},
  publisher = {American Physical Society},
  doi = {10.1103/PhysRevB.78.155117},
  url = {https://link.aps.org/doi/10.1103/PhysRevB.78.155117}
}

@article{PhysRevLett.101.090603,
  title = {Accurate Determination of Tensor Network State of Quantum Lattice Models in Two Dimensions},
  author = {Jiang, H. C. and Weng, Z. Y. and Xiang, T.},
  journal = {Phys. Rev. Lett.},
  volume = {101},
  issue = {9},
  pages = {090603},
  numpages = {4},
  year = {2008},
  month = {Aug},
  publisher = {American Physical Society},
  doi = {10.1103/PhysRevLett.101.090603},
  url = {https://link.aps.org/doi/10.1103/PhysRevLett.101.090603}
}

@article{PhysRevB.99.195105,
  title = {Universal tensor-network algorithm for any infinite lattice},
  author = {Jahromi, Saeed S. and Or\'us, Rom\'an},
  journal = {Phys. Rev. B},
  volume = {99},
  issue = {19},
  pages = {195105},
  numpages = {15},
  year = {2019},
  month = {May},
  publisher = {American Physical Society},
  doi = {10.1103/PhysRevB.99.195105},
  url = {https://link.aps.org/doi/10.1103/PhysRevB.99.195105}
}

@misc{qiskit_paper,
      title={Quantum computing with Qiskit}, 
      author={Ali Javadi-Abhari and Matthew Treinish and Kevin Krsulich and Christopher J. Wood and Jake Lishman and Julien Gacon and Simon Martiel and Paul D. Nation and Lev S. Bishop and Andrew W. Cross and Blake R. Johnson and Jay M. Gambetta},
      year={2024},
      eprint={2405.08810},
      archivePrefix={arXiv},
      primaryClass={quant-ph},
      url={https://arxiv.org/abs/2405.08810}, 
}

@inproceedings{RCM_graph_algo,
author = {Cuthill, E. and McKee, J.},
title = {Reducing the bandwidth of sparse symmetric matrices},
year = {1969},
isbn = {9781450374934},
publisher = {Association for Computing Machinery},
address = {New York, NY, USA},
url = {https://doi.org/10.1145/800195.805928},
doi = {10.1145/800195.805928},
booktitle = {Proceedings of the 1969 24th National Conference},
pages = {157-172},
numpages = {16},
series = {ACM '69}
}

@article{Hangleiter:2022ibu,
    author = "Hangleiter, Dominik and Eisert, Jens",
    title = "{Computational advantage of quantum random sampling}",
    eprint = "2206.04079",
    archivePrefix = "arXiv",
    primaryClass = "quant-ph",
    doi = "10.1103/RevModPhys.95.035001",
    journal = "Rev. Mod. Phys.",
    volume = "95",
    number = "3",
    pages = "035001",
    year = "2023"
}

@article{deshpande2025peaked,
  title={Peaked quantum advantage using error correction},
  author={Deshpande, Abhinav and Fefferman, Bill and Ghosh, Soumik and Gullans, Michael and Hangleiter, Dominik},
  journal={arXiv preprint arXiv:2510.05262},
  year={2025}
}

@book{chen2016report,
  title={Report on post-quantum cryptography},
  author={Chen, Lily and Chen, Lily and Jordan, Stephen and Liu, Yi-Kai and Moody, Dustin and Peralta, Rene and Perlner, Ray A and Smith-Tone, Daniel},
  volume={12},
  year={2016},
  publisher={US Department of Commerce, National Institute of Standards and Technology~…}
}

@article{Bouland:2018bva,
    author = "Bouland, Adam and Fefferman, Bill and Nirkhe, Chinmay and Vazirani, Umesh",
    title = "{On the complexity and verification of quantum random circuit sampling}",
    eprint = "1803.04402",
    archivePrefix = "arXiv",
    primaryClass = "quant-ph",
    doi = "10.1038/s41567-018-0318-2",
    journal = "Nature Phys.",
    volume = "15",
    number = "2",
    pages = "159--163",
    year = "2018"
}

@article{Zhang:2025tuf,
    author = "Zhang, Yuxuan",
    title = "{Complexity and hardness of random peaked circuits}",
    eprint = "2510.00132",
    archivePrefix = "arXiv",
    primaryClass = "quant-ph",
    month = "9",
    year = "2025"
}

@article{Aaronson:2016guw,
    author = "Aaronson, Scott and Chen, Lijie",
    title = "{Complexity-Theoretic Foundations of Quantum Supremacy Experiments}",
    eprint = "1612.05903",
    archivePrefix = "arXiv",
    primaryClass = "quant-ph",
    month = "12",
    year = "2016"
}

@article{Mori:2023ceu,
    author = "Mori, Yusei and Hakoshima, Hideaki and Sudo, Kyohei and Mori, Toshio and Mitarai, Kosuke and Fujii, Keisuke",
    title = "{Quantum circuit unoptimization}",
    eprint = "2311.03805",
    archivePrefix = "arXiv",
    primaryClass = "quant-ph",
    doi = "10.1103/PhysRevResearch.7.023139",
    journal = "Phys. Rev. Res.",
    volume = "7",
    number = "2",
    pages = "023139",
    year = "2025"
}

@article{doi:10.1142/S0219749905001067,
author = {Janzing, Dominik and Wocjan, Pawel and Beth, Thomas},
title = {"Non-Identity-Check" is QMA-Complete},
journal = {International Journal of Quantum Information},
volume = {03},
number = {03},
pages = {463-473},
year = {2005},
doi = {10.1142/S0219749905001067},

URL = { 
    
        https://doi.org/10.1142/S0219749905001067
    
    

},
eprint = { 
    
        https://doi.org/10.1142/S0219749905001067
    
    

}
}

@misc{ji2009nonidentitycheckremainsqmacomplete,
      title={Non-Identity Check Remains QMA-Complete for Short Circuits}, 
      author={Zhengfeng Ji and Xiaodi Wu},
      year={2009},
      eprint={0906.5416},
      archivePrefix={arXiv},
      primaryClass={quant-ph},
      url={https://arxiv.org/abs/0906.5416}, 
}

@Inbook{Kaliski2005,
author="Kaliski, Burt",
editor="van Tilborg, Henk C. A.",
title="RSA factoring challenge",
bookTitle="Encyclopedia of Cryptography and Security",
year="2005",
publisher="Springer US",
address="Boston, MA",
pages="531--532",
isbn="978-0-387-23483-0",
doi="10.1007/0-387-23483-7_362",
url="https://doi.org/10.1007/0-387-23483-7_362"
}

\end{document}